\documentclass[12pt]{article}
\usepackage{amsmath}
\usepackage{graphicx,psfrag,epsf}
\usepackage{enumerate}
\usepackage{natbib}

\usepackage{amsthm}
\usepackage{amssymb}
\usepackage{bm}
\usepackage{dsfont}
\usepackage{caption}
\usepackage[labelformat=simple]{subcaption}

\usepackage{tabularx}
\usepackage{tabu}
\usepackage{pdfpages}
\usepackage{array}
\usepackage{longtable}
\usepackage{pdflscape}

\usepackage{url}
\DeclareMathOperator*{\argmax}{arg\,max}

\usepackage{tikz}
\tikzset{
  treenode/.style = {shape = circle, align = center, inner sep = 0pt, minimum size = 0.2cm},
  root/.style     = {treenode},
  env/.style      = {treenode, font = \normalsize},
  dummy/.style    = {treenode}
}

\definecolor{deep-gray}{gray}{0.25}

\newcommand{\beginsupplement}{%
\setcounter{table}{0}
\renewcommand{\thetable}{S\arabic{table}}%
\setcounter{figure}{0}
\renewcommand{\thefigure}{S\arabic{figure}}%
\setcounter{subsection}{0}
\renewcommand{\thesubsection}{S\arabic{subsection}}%
\setcounter{equation}{0}
\renewcommand{\theequation}{S\arabic{equation}}%
}

\theoremstyle{plain}
\newtheorem{theorem}{Theorem}
\newtheorem{lemma}{Lemma}
\newtheorem{corollary}{Corollary}
\theoremstyle{definition}
\newtheorem{definition}{Definition}
\theoremstyle{remark}
\newtheorem{remark}{Remark}

\newcommand{\blind}{0}

\def\sobm {\mathcal{Z}}
\def\zsobm {\mathcal{Z}^0}
\def\tsobm {\tilde{\mathcal{Z}}}
\def\bx {\bm x}
\def\bz {\bm z}
\def\bpi {\bm \pi}

\def\tree {\mathcal{T}}
\newcommand{\ie}{\textit{i.e.}}
\newcommand{\eg}{\textit{e.g.}}

\begin{document}

\def\spacingset#1{\renewcommand{\baselinestretch}%
{#1}\small\normalsize} \spacingset{1}


\if0\blind
{
  \title{\bf Posterior Contraction Rate of Sparse Latent Feature Models with Application to Proteomics}
  \author{Tong Li\thanks{Department of Statistics, Columbia University} , Tianjian Zhou\thanks{Department of Public Health Sciences, University of Chicago} , Kam-Wah Tsui\thanks{Department of Statistics, University of Wisconsin--Madison} , Lin Wei\thanks{Research Institute, NorthShore University HealthSystem} , and Yuan Ji\footnotemark[2] }
  \maketitle
} \fi

\if1\blind
{
  \bigskip
  \bigskip
  \bigskip
  \begin{center}
    {\LARGE\bf Title}
\end{center}
  \medskip
} \fi

\bigskip
\begin{abstract}
The Indian buffet process (IBP) and phylogenetic Indian buffet process (pIBP) can be used as prior models to infer latent features in a data set. The theoretical properties of these models are under-explored, however, especially in high dimensional settings. In this paper, we show that under mild sparsity condition, the posterior distribution of the latent feature matrix, generated via IBP or pIBP priors, converges to the true latent feature matrix asymptotically. We derive the posterior convergence rate, referred to as the contraction rate. We show that the convergence holds even when the dimensionality of the latent feature matrix increases with the sample size, therefore making the posterior inference valid in high dimensional setting. We demonstrate the theoretical results using computer simulation, in which the parallel-tempering Markov chain Monte Carlo method is applied to overcome computational hurdles. The practical utility of the derived properties is demonstrated by inferring the latent features in a reverse phase protein arrays (RPPA) dataset under the IBP prior model. Software and dataset reported in the manuscript are provided at \url{http://www.compgenome.org/IBP}.
\end{abstract}

\noindent%

\spacingset{1.45}

\section{Introduction}

The latent feature models are concerned about finding latent structures in a data set $X_{n \times p}$ where each row $\bx_i=(x_{i1}, \cdots, x_{ip})$ represents a single observation of $p$ objects and $n$ is the sample size. 
We consider the case where the number of objects $p = p_n$ increases as sample size $n$ increases.
The goal is to explain the variability of the observed data with a latent binary feature matrix $Z_{p \times K}$ where each column of $Z$ represents a latent feature that includes a subset of the $p$ objects. The number of latent features $K$ is unknown and is inferred as well.

Bayesian nonparametric latent feature models such as the Indian buffet process (IBP) \citep{griffiths2006infinite, griffiths2011indian} can be used to define the prior distribution of the binary latent feature matrix with arbitrarily many columns. In many applications (such as \citealt{chu2006identifying}) these priors could lead to desirable posterior inference. An important property of IBP is that the corresponding distribution maintains exchangeability across the rows that index the experimental units, making posterior inference relatively simple and easy to implement.
However, sometimes the rows of the latent feature matrix must follow a group structure, such as in phylogenetic inferences. To address such needs, 
the phylogenetic Indian Buffet Process (pIBP) \citep{miller2012phylogenetic} has been developed to allow
different rows to be partially exchangeable.

Despite the increasing popularity in the application of IBP and pIBP prior models, such as in cancer and evolutional genomics, few 
theoretical results have been discussed on the posterior inference based on these models. For example, from a frequentist view, it is important to investigate the  asymptotic convergence of the posterior distribution of the latent feature matrix under IBP and pIBP priors.  Existing literature on the theory of Bayesian posterior consistency includes, for example, \cite{schwartz1965bayes, barron1999consistency} and \cite{ghosal2000convergence}. 
\cite{chen2016posterior}  is a motivational work exploring  theoretical properties of the posterior distribution of the latent feature matrix based on IBP or pIBP priors. They explored the asymptotic behavior of the IBP or pIBP-based posterior inference, where the sample size $n$ increases in a much faster speed than 
the number of objects $p_n$, i.e., the dimensionality of $Z$.
This might be hard to achieve in some real applications. We consider important extensions based on \cite{chen2016posterior}. In particular, we consider properties of posterior inference based on IBP and pIBP priors in high dimensions and with sparsity. 
Under a similar high dimensional and sparse setting, a related work is \cite{pati2014posterior}, where the authors studied the asymptotic behavior of sparse Bayesian factor models discussed in \cite{west2003bayesian}. These models are concerned about continuous latent features, which are different from the binary feature models like IBP and pIBP.

High dimensional inference is now routinely needed in many applications, such as genomics and proteomics. Due to the reduced cost of high-throughput  biological experiments (e.g., next-generation sequencing), a number of genomics elements (such as genes) can be measured with a relatively short amount of time and low cost for a large number of patients. In our application, the number of genomics elements $n$ is the sample size, the number of patients $p$ is the number of rows, or the number of objects in the latent feature matrix, and the number of latent features $K$ is assumed unknown. Depending on the particular research question, in some applications, genes can be the objects and patients can be the samples.
When $p = p_n$ becomes large relative to $n$, the sparsity of the feature matrix critically ensures the efficiency and validity of statistical inference. 
We will show that under the sparsity condition, the requirement for posterior convergence can be relaxed from $p_n^3=o(n)$ \citep{chen2016posterior}  to $p_n(\log p_n)^2 = o(n)$ (see Remark \ref{rem-main}).
 
The proposed  sparsity condition can be reasonably interpreted
and checked in practice. 
For example, in genomics and proteomics applications, our
sparsity condition means that the number of features shared by different patients is small, \ie, the patients are heterogeneous.
This is different from some published sparsity conditions that involve
more complicated mathematical expressions, possibly in terms of the
properties of complex matrices, which are  difficult to check in
real-world applications. 

The rest of the paper is organized as follows. Section \ref{sec:framework} introduces the latent feature model and the IBP/pIBP priors. Section \ref{sec:rate} establishes the posterior contraction rate of sparse latent feature models under IBP and pIBP prior, which is the main theoretical result of this paper. Section \ref{sec:MCMC} proposes an efficient posterior inference scheme based on Markov chain Monte Carlo (MCMC) simulations. Section \ref{sec:exap} provides both simulated and real-world proteomics examples that support the theoretical derivations.
We conclude the paper with a brief discussion in Section \ref{sec:discussion}. Some technical details are provided in the supplement.

\section{Notation and Probability Framework}
\label{sec:framework}
In this section, we first introduce some notation, 
and then specify the hierarchical model including the sampling model and the prior model.
In particular, the sampling model is the latent feature model, and the prior model is the IBP mixture or pIBP mixture.

\subsection{Notation}
Throughout the paper, we denote by $p(\cdot)$ and $P(\cdot)$ probability density functions (pdf) and probability mass functions (pmf), respectively. Specifically for the latent feature matrix $Z$, we use $\Pi(Z)$ and $\Pi(Z \mid X)$ to denote the prior and posterior distribution of $Z$, respectively. The likelihood $p(X \mid Z) = \prod_{i = 1}^n p(\bx_i^T \mid Z)$. For two sequences $a_n$ and $b_n$, the notation $a_n = O(b_n)$ means there exists a positive real number $C$ and a constant $n_0$ such that $a_n \leq C b_n$ for all $n \geq n_0$; the notation $a_n = o(b_n)$ means for every positive real number $\varepsilon$ there exists a constant $n_0$ such that $a_n \leq \varepsilon b_n$ for all $n \geq n_0$. For a matrix $A$, $\| A \|$ denotes the spectral norm defined as the largest singular value of $A$. Finally, $C$ is a generic notation for positive constants whose value might change depending on the context but is independent from other quantities.

\subsection{Latent Feature Model}
Suppose that $X_{n\times p}$ is a collection of the observed data. Each row $\bx_i=(x_{i1},\cdots,x_{ip})$ represents a single observation of $p$ objects, for $i=1, \cdots, n$, where $\bx_i$'s are independent. Assume that the mechanism of generating $X$ can be characterized by latent features, 
\begin{align}
X^T = Z A + E.
\label{eq:lfm}
\end{align}
Here $Z = (z_{jk})_{p \times K}$ denotes the latent binary feature matrix, each entry $z_{jk} \in \{0,1\}$ represents object $j$ possesses feature $k$ ($z_{jk} = 1$) or not ($z_{jk} = 0$), respectively. The loading matrix $A = (a_{ki})_{K \times n}$, with each entry being the contribution of the $k$-feature to the $i$-th observation. We assume $a_{ki} \stackrel{i.i.d}{\sim} N(0,\sigma_a^2)$.
The error matrix $E = (e_{ji})_{p \times n}$, where $e_{ji}$'s are independent Gaussian errors, $e_{ji} \stackrel{i.i.d}{\sim} N(0,\sigma^2)$.

After integrating out $A$, we obtain for each observation of the $p$ objects,
\begin{align*}
\bx_i^{T} \stackrel{i.i.d}{\sim} N_p(\bm 0, \sigma_a^2 ZZ^T+\sigma^2 \bm{I}_p),
\end{align*}
where $N_p$ is a $p$-variate Gaussian distribution.
Therefore, the conditional distribution of $X$ given $Z$ is
\begin{align}
p(X|Z)=\prod\limits_{i=1}^n\frac{1}{\sqrt{(2\pi)^p\det\left( \sigma_a^2 ZZ^T + \sigma^2 \bm{I}_p\right)}}\exp\left(-\frac{1}{2}\bx_i\left(\sigma_a^2 ZZ^T + \sigma^2 \bm{I}_p\right)^{-1}\bx_i^T\right).
\label{eq:sampling_model}
\end{align}

Without loss of generality, we always assume that $\sigma^2=\sigma_a^2=1$.
One of the primary interests is to conduct appropriate estimation on $ZZ^T$, which is usually called the \emph{similarity matrix} since each entry of $ZZ^T$ is the number of features shared by two objects.

\subsection{Prior Distributions Based on IBP and pIBP}

In the latent feature model (Equation \ref{eq:lfm}), it remains to specify the prior for the binary feature matrix $Z$. IBP and pIBP are popular prior choices on binary matrices with an unbounded number of columns. 
IBP assumes exchangeability among the objects, while
pIBP introduces dependency among the entries of the $k$-th column of $Z$ through a rooted tree $\tree$. See Figure \ref{fig:eg_tree} for an example of the tree. IBP is a special case of pIBP when the root node is the only internal node of the tree.
The construction and the pmf of IBP are described and derived in detail in \cite{griffiths2011indian}. For pIBP, only a brief definition is given in \cite{miller2012phylogenetic}. 
For the proof of the main theoretical result of this paper, we propose a construction of pIBP in a similar way as IBP and derive the pmf of pIBP.

We first introduce some notation.
Let $\sobm_{p \times K}$ denote the collection of binary matrices with $p$ rows and $K$ ($ K \in \mathbb{N}^+$) columns such that none of these columns consist of all 0's. 
Let $\sobm_{p} = \cup_{K = 1}^{\infty} \sobm_{p \times K}$ and let $\mathbf{0}$ denote a $p$-dimensional vector whose elements are all 0's, \ie, 
$$
\mathbf{0}=(0,0,\cdots,0)^T.
$$
In the following sections, we also regard $\mathbf{0}$ as a $p\times 1$ matrix when needed.  Both IBP and pIBP are defined over $\zsobm_p \triangleq \{\mathbf{0}\}\bigcup \sobm_p$. 
It can be shown that with probability 1, a draw from IBP or pIBP has only finitely many columns. 
For the construction of IBP and pIBP, we introduce some more notations as follows. 
Denote by $\tsobm_{p \times \tilde{K}}$ the collection of all binary matrices with $p$ rows and $\tilde{K}$ columns (where the columns can be all zeros). 
We define a many-to-one mapping $G( \cdot ) : \tsobm_{p \times \tilde{K}} \rightarrow \zsobm_{p}$. For a binary matrix $\tilde{Z} \in \tsobm_{p \times \tilde{K}}$, if all columns of $\tilde{Z}$ are $\bm 0$'s, then $G(\tilde{Z}) = \bm 0$; otherwise, $G(\tilde{Z})$ is obtained by deleting all zero columns of $\tilde{Z}$. For the purpose of performing inference on the similarity matrix $ZZ^T$, it suffices to focus on the set of equivalence classes induced by $G(\cdot)$. Two matrices $\tilde{Z}_1, \tilde{Z}_2 \in \tsobm_{p \times \tilde{K}}$ are $G$-equivalent if $G(\tilde{Z}_1) = G(\tilde{Z}_2)$, and in this case the similarity matrices induced by $\tilde{Z}_1$ and $\tilde{Z}_2$ are the same, $\tilde{Z}_1 \tilde{Z}_1^T = \tilde{Z}_2 \tilde{Z}_2^T$.

We now turn to a constructive definition of pIBP. The pIBP can be constructed in the following three steps by taking the limit of a finite feature model.

\noindent\textbf{Step 1.}
Given some hyperparameter $\alpha$ and tree structure $\tree$, we start by defining a probability distribution $P_{\tilde{K}}(\tilde{Z} \mid \alpha, \tree)$ over $\tilde{Z} \in \tsobm_{p \times \tilde{K}}$. Denote by $\tilde{Z} = (\tilde{\bm z}_1, \tilde{\bm z}_2,\cdots,\tilde{\bm z}_{\tilde{K}})$.
Let $\bpi=\{\pi_1,\pi_2,\cdots,\pi_{\tilde{K}}\}$ be a vector of success probabilities such that $(\pi_k \mid \alpha) \stackrel{i.i.d.}{\sim} \text{Beta} ( \alpha/ \tilde{K}, 1 )$. The columns of $\tilde{Z}$ are conditionally independent given $\bpi$ and $\tree$,
\begin{align*}
P_{\tilde{K}}(\tilde{Z}  \mid  \bpi, \tree) = \prod \limits_{k=1}^{\tilde{K}}P(\tilde{\bm z}_k \mid \pi_k, \tree),
\end{align*}
where $P(\tilde{\bm z}_k \mid \pi_k, \tree)$ is determined as in \cite{miller2012phylogenetic} (more details in Supplementary Section \ref{app:deriv_pibp}). 
If $P(\tilde{\bm z}_k \mid \pi_k, \tree) = \prod_{j = 1}^p P(\tilde{z}_{jk} \mid \pi_k)$ where $\tilde{z}_{jk} \stackrel{i.i.d}{\sim} \text{Bernoulli}(\pi_k)$, pIBP reduces to IBP.
The marginal probability of $\tilde{Z}$ is
\begin{align*}
P_{\tilde{K}}(\tilde{Z}  \mid  \alpha, \tree) = \prod \limits_{k=1}^{\tilde{K}} \int P(\tilde{\bm z}_k \mid \pi_k, \tree) p(\pi_k \mid \alpha) d \pi_k.
\end{align*}

\noindent\textbf{Step 2.} 
Next, for any $Z \in \zsobm_p$ with $K$ columns, we define a probability distribution (for $\tilde{K} \geq K$)
\begin{align*}
\Pi_{\tilde{K}}(Z  \mid  \alpha, \tree) \triangleq \sum_{\tilde{Z} \in \tsobm_{p \times \tilde{K}}: G(\tilde{Z}) = Z} P_{\tilde{K}} (\tilde{Z} \mid  \alpha, \tree),
\end{align*}
for $P_{\tilde{K}}(\tilde{Z}  \mid  \alpha, \tree)$ defined in Step 1. That is, we collapse all binary matrices in $\tsobm_{p \times \tilde{K}}$ that are $G$-equivalent.

\noindent\textbf{Step 3.}
Finally, for any $Z \in \zsobm_p$, define 
\begin{align*}
\Pi(Z \mid \alpha, \tree) \triangleq \lim\limits_{  \tilde{K}\rightarrow \infty} \Pi_{\tilde{K}}  (Z \mid \alpha, \tree).
\end{align*}
Here $\Pi(Z \mid \alpha, \tree)$ is the pmf of pIBP under $G$-equivalence classes.

Based on the three steps of constructing pIBP, we derive the pmf of pIBP given $\alpha$ and $\tree$.
Details on the derivation is given in Supplementary Section \ref{app:deriv_pibp}. 
Let $S(\tree)$ denote the total edge lengths of the tree structure (see Figure \ref{fig:eg_tree}) and $\psi(\cdot)$ denote the digamma function. For $Z\in \zsobm_p$, we have
\begin{align}
\Pi(Z \mid \alpha, \tree)=\begin{cases}
\exp \left\{-\left( \psi \left( S(\tree)+1\right) -\psi(1)\right) \alpha \right\}, & \text{if } Z = \bm{0};\\
\exp\left\lbrace -\left( \psi\left( S(\tree)+1\right) -\psi(1)\right) \alpha\right\rbrace \ \frac{\alpha^{K}}{K!}\ \prod\limits_{k=1}\limits^{K}\lambda_k, & \text{if } Z \in \sobm_{p \times K},
\end{cases}
\label{eq:pmf-pIBP}
\end{align}
where $\lambda_k \triangleq \lambda(\bm{z}_k,\tree) = \int_0^1 P(\bm{z}_k \mid \pi_k, \tree) \pi^{-1}_k d \pi_k$ (See Supplementary Section \ref{app:deriv_pibp}) and $\bm{z}_k$ is the $k$-th column of $Z$. 

Assume $\alpha \sim \text{Gamma}(1, 1)$. After integrating out $\alpha$, we obtain the pmf of pIBP mixture,
\begin{align*}
\Pi(Z \mid \tree)=\begin{cases}
\Big(\psi(S(\tree)+1)-\psi(1)+1\Big) ^{-1}, & \text{if } Z = \mathbf{0};\\
\Big(\psi(S(\tree)+1)-\psi(1)+1\Big) ^{-(K+1)}\prod\limits_{k=1}\limits^{K}\lambda_k, & \text{if } Z \in \sobm_{p \times K}.
\end{cases}
\end{align*}
For notational simplicity, we suppress the condition on $\tree$ hereafter when we discuss pIBP. 

When $S(\tree) = p$ and $\lambda_k = \int_0^1 \prod_{j = 1}^p \pi_k^{m_k - 1} (1 - \pi_k)^{p - m_k} d \pi_k$, pIBP reduces to IBP, where $m_k=\sum_{j=1}^{p}z_{jk}$ denotes the number of objects possessing feature $k$. The pmf of IBP mixture is
\begin{align*}
\Pi(Z) =\begin{cases}
(H_p+1)^{-1}, & \text{if } Z=\mathbf{0};\\
(H_p+1)^{-(K+1)}\prod_{k=1}^{K}\frac{(p-m_k)! (m_k-1)!}{p!}, & \text{if } Z \in \sobm_{p \times K},
\end{cases}
\end{align*}
where $H_p=\sum_{j=1}^{p} (1 / j)$.

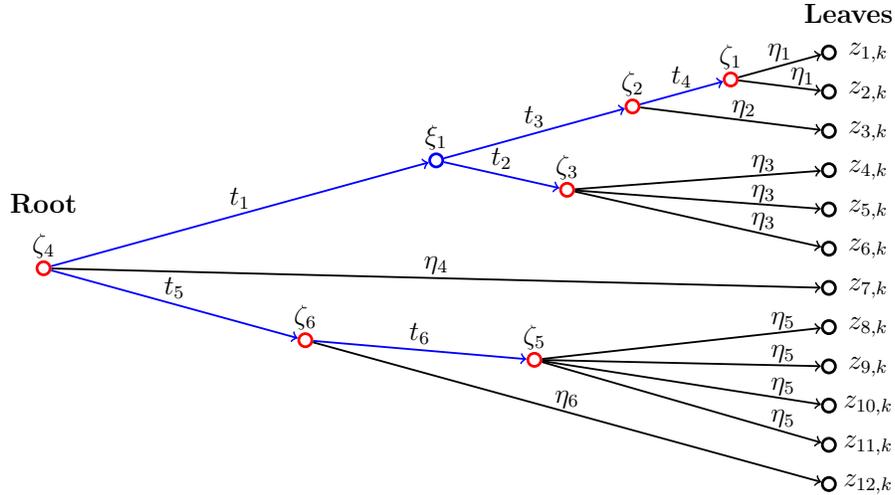
\begin{figure}[h!]
\begin{center}
\scalebox{0.87}{
\begin{tikzpicture}
[every node/.style = {very thick}]
    \node[root, draw = red] at (0, 0)         (root) {};
    
    \node[dummy, draw = blue] at (6, 1.65)   (xi1) {};
    \node[dummy, draw = red] at (10.5, 2.8875)   (zeta1) {};
    \node[dummy, draw = red] at (9, 2.475)   (zeta2) {};
    \node[dummy, draw = red] at (8, 1.2)   (zeta3) {};
    \node[dummy, draw = red] at (7.5, -1.4)   (zeta5) {};
    \node[dummy, draw = red] at (4, -1.1)   (zeta6) {};
    
	\node[dummy, draw = black] at (12, 3.3)   (z1) {};
    \node[dummy, draw = black] at (12, 2.7)   (z2) {};
    \node[dummy, draw = black] at (12, 2.1)   (z3) {};
    \node[dummy, draw = black] at (12, 1.5)   (z4) {};
    \node[dummy, draw = black] at (12, 0.9)   (z5) {};
    \node[dummy, draw = black] at (12, 0.3)   (z6) {};
    \node[dummy, draw = black] at (12, -0.3)   (z7) {};
    \node[dummy, draw = black] at (12, -0.9)   (z8) {};
    \node[dummy, draw = black] at (12, -1.5)   (z9) {};
    \node[dummy, draw = black] at (12, -2.1)   (z10) {};
    \node[dummy, draw = black] at (12, -2.7)   (z11) {};
    \node[dummy, draw = black] at (12, -3.3)   (z12) {};    
    
    \node[env] at (0, 1)    (root_text) {\textbf{Root}};
    \node[env] at (10.5, 3.2375)    (zeta1_rv) {$\zeta_1$};
    \node[env] at (9, 2.825)    (zeta2_rv) {$\zeta_2$};
    \node[env] at (8, 1.55)    (zeta3_rv) {$\zeta_3$};
    \node[env] at (0, 0.35)    (root_rv) {$\zeta_4$};
    \node[env] at (7.5, -1.05) (zeta5_rv) {$\zeta_5$};
    \node[env] at (4, -0.75) (zeta6_rv) {$\zeta_6$};
    \node[env] at (6, 2) (xi1_rv) {$\xi_1$};
    
    \node[env] at (12.3, 3.9) (leaves_text) {\textbf{Leaves}};
    \node[env] at (12.6, 3.3) (z1_rv) {$z_{1, k}$};
    \node[env] at (12.6, 2.7) (z2_rv) {$z_{2, k}$};
    \node[env] at (12.6, 2.1) (z3_rv) {$z_{3, k}$};
    \node[env] at (12.6, 1.5) (z4_rv) {$z_{4, k}$};
    \node[env] at (12.6, 0.9) (z5_rv) {$z_{5, k}$};
    \node[env] at (12.6, 0.3) (z6_rv) {$z_{6, k}$};
    \node[env] at (12.6, -0.3) (z7_rv) {$z_{7, k}$};
    \node[env] at (12.6, -0.9) (z8_rv) {$z_{8, k}$};
    \node[env] at (12.6, -1.5) (z9_rv) {$z_{9, k}$};
    \node[env] at (12.6, -2.1) (z10_rv) {$z_{10, k}$};
    \node[env] at (12.6, -2.7) (z11_rv) {$z_{11, k}$};
    \node[env] at (12.6, -3.3) (z12_rv) {$z_{12, k}$};
    
    \node[env] at (3, 1.075) (edge1) {$t_1$};
    \node[env] at (7, 1.675) (edge2) {$t_2$};
    \node[env] at (7.5, 2.3125) (edge3) {$t_3$};
    \node[env] at (9.75, 2.93125) (edge4) {$t_4$};
    \node[env] at (2, -0.3) (edge5) {$t_5$};
    \node[env] at (5.75, -1) (edge6) {$t_6$};
    
    \node[env] at (11.25, 3.3) (eta1_1) {$\eta_1$};
    \node[env] at (11.6, 2.93) (eta1_2) {$\eta_1$};
    \node[env] at (10.7, 2.4375) (eta2) {$\eta_2$};
    
    \node[env] at (11, 1.6) (eta3_1) {$\eta_3$};
    \node[env] at (11, 1.15) (eta3_2) {$\eta_3$};
    \node[env] at (11, 0.72) (eta3_3) {$\eta_3$};
    
    \node[env] at (6, 0.05) (eta4) {$\eta_4$};
    
    \node[env] at (11.3, -0.8) (eta5_1) {$\eta_5$};
    \node[env] at (11.3, -1.3) (eta5_2) {$\eta_5$};
    \node[env] at (11.3, -1.8) (eta5_3) {$\eta_5$};
    \node[env] at (11.3, -2.3) (eta5_4) {$\eta_5$};
    
    \node[env] at (8, -2) (eta6) {$\eta_6$};

    \draw[draw = blue, ->, line width = 0.8pt] (root) -- (xi1);
    \draw[draw = blue, ->, line width = 0.8pt] (xi1) -- (zeta2);
    \draw[draw = blue, ->, line width = 0.8pt] (zeta2) -- (zeta1);
    \draw[->, line width = 0.8pt] (zeta1) -- (z1);
    \draw[->, line width = 0.8pt] (zeta1) -- (z2);
    \draw[->, line width = 0.8pt] (zeta2) -- (z3);
    \draw[draw = blue, ->, line width = 0.8pt] (xi1) -- (zeta3);
    \draw[->, line width = 0.8pt] (zeta3) -- (z4);
    \draw[->, line width = 0.8pt] (zeta3) -- (z5);
    \draw[->, line width = 0.8pt] (zeta3) -- (z6);
    \draw[->, line width = 0.8pt] (root) -- (z7);
    \draw[->, line width = 0.8pt] (zeta5) -- (z8);
    \draw[->, line width = 0.8pt] (zeta5) -- (z9);
    \draw[->, line width = 0.8pt] (zeta5) -- (z10);
    \draw[->, line width = 0.8pt] (zeta5) -- (z11);
    \draw[draw = blue, ->, line width = 0.8pt] (zeta6) -- (zeta5);
    \draw[draw = blue, ->, line width = 0.8pt] (root) -- (zeta6);
    \draw[->, line width = 0.8pt] (zeta6) -- (z12);
\end{tikzpicture}
}
\end{center}
\caption{An example of a tree structure $\tree$, which is a directed graph with random variables at the nodes (marked as circles). 
Entries of the $k$-th column of $Z$, $z_{jk}$'s, are at the leaves.
The lengths of all edges of $\tree$, $t_i$'s and $\eta_l$'s, are marked on the figure. In particular, $\eta_l$'s represent the lengths between each leaf ($z_{jk}$, in black) and its parent node ($\zeta_l$, in red).
The total edge lengths $S(\tree)$ is the summation of the lengths of all edges of $\tree$.
In this example, $S(\tree) = \sum_{1\leq i\leq 6}t_i+(2\eta_1+\eta_2+3\eta_3+\eta_4+4\eta_5+\eta_6)$. The condition in case (2) of Lemma \ref{lem1} in Section \ref{sec:rate} means $\inf\limits_{1\leq l\leq 6}\eta_l\geq \eta_0$ for some $\eta_0 > 0$.
}
\label{fig:eg_tree}

\end{figure}

\section{Posterior Contraction Rate under the Sparsity Condition}
\label{sec:rate}

In this section, we establish the posterior contraction rate of IBP mixture and pIBP mixture under a sparsity condition. All the proofs are given in the supplement.

The sparsity condition is defined below for a sequence of binary matrices $\{ Z_n^*, n = 1, 2, \ldots\}$.

\begin{definition}[\textbf{Sparsity}] \label{def:sparsity}
\theoremstyle{definition}
Consider a sequence of binary matrices $\{ Z_{n}^*, n = 1, 2, \ldots \}$
where $Z_{n}^*=(z_{jk}^*)_{p_n \times K_{n}^*}$. 
Assume that $m_{kn}\triangleq \sum_{j=1}^{p_n} z_{jk}^* \leq s_n$  for some $s_n\geq 1$ and all $k = 1, \ldots, K_{n}^*$. We say $\{ Z_{n}^*, n = 1, 2, \ldots \}$ are sparse if $s_n / p_n \rightarrow 0$ as $n\rightarrow \infty$.
\end{definition}

The condition indicates that as sample size increases, the number of objects possessing any feature is upper-bounded and must be relatively small compared to the total number of objects.  
Therefore, such an assumption may be  
 assessed via simulation studies and then applied to real-world
 applications. Examples will be provided later on. 
 Similar but more strict assumptions are made in \cite{castillo2012needles} and \cite{pati2014posterior}, under different contexts.

Next, we review the definition of posterior contraction rate.

\begin{definition}[\textbf{Posterior Contraction Rate}]\label{def2}
Let $\{ Z_{n}^*, n = 1, 2, \ldots \}$ represent a sequence of true latent feature matrices where each $Z_{n}^*$ has $p_n$ rows. For each $n$, the observations are generated from $\bx_i^{T}\stackrel{i.i.d.}{\sim} N(0,Z_{n}^* Z_{n}^{*T}+\bm{I}_{p_n})$ for $i=1,2,\cdots,n$.
Denote by $\Pi(\cdot \mid X)$ the posterior distribution of the latent feature matrix under IBP mixture or pIBP mixture prior. If
\begin{align*}
E_{Z_{n}^*} \left[ \Pi \left( \|Z_nZ_n^{T}-Z_{n}^* Z_{n}^{*T}\|\leq C\epsilon_n \mid X \right) \right] \rightarrow 1
\end{align*}
as $n\rightarrow \infty$, where $\|\cdot\|$ represents the spectral norm and $C$ is a positive constant, then we say the posterior contraction rate of $Z_nZ_n^{T}$ to the true $Z_{n}^* Z_{n}^{*T}$ is $\epsilon_n$ under the spectral norm.
\end{definition}

For the proof of the main theorem, we derive the following lemma.  The lemma establishes the lower bound of $\Pi(Z_n = Z_{n}^*)$ for a sequence of binary matrices $\{ Z_{n}^*, n = 1, 2, \ldots \}$ satisfying the sparsity condition, where $\Pi(\cdot)$ represents the pmf of IBP mixture or pIBP mixture. 

\begin{lemma}\label{lem1}
Consider a sequence of binary matrices $\{ Z_{n}^*: Z_{n}^* \in \{\mathbf{0}\} \bigcup \sobm_{p_n \times K_{n}^{*}}, n = 1, 2, \ldots \}$ that are sparse under Definition \ref{def:sparsity}. Parameters $m_{kn}$ and $s_n$ are defined accordingly (for $\mathbf{0}$ matrices, let $s_n=K_n^*=1$). We have
\begin{align*}
\Pi(Z_n = Z_{n}^* )\geq \exp \left( -Cs_nK_{n}^{*}\log(p_n+1) \right)
\end{align*}
for some positive constant $C$, if either of the following two cases is true: (1) $Z_n$ follows IBP mixture; (2) $Z_n$ follows pIBP mixture, and the minimal length between each leaf and its parent node is lower bounded by $\eta_0 > 0$ (see Figure \ref{fig:eg_tree}).
\end{lemma}

\begin{remark}
\label{rem1}
Results in Lemma \ref{lem1} depend on the sparsity condition. As a counterexample, we approximate $\Pi(Z_n = Z_{n}^*)$ for a sequence of non-sparse binary matrices $\{Z_{n}^*, n = 1, 2, \ldots \}$, where $Z_{n}^* \in \sobm_{p_n}$, and $Z_n$ follows IBP mixture. Recall that for IBP mixture, we have 
\begin{align*}
\Pi(Z_n = Z_{n}^* )=\frac{1}{(H_{p_n}+1)^{K_{n}^{*}+1}}\prod_{k=1}^{K_{n}^{*}}\frac{(p_n-m_{kn})! (m_{kn}-1)!}{p_n!}.
\end{align*}
Let $m_{kn}= p_n / 2$ for every column $k = 1, \ldots, K_{n}^{*}$, then $[(p_n-m_{kn})!(m_{kn}-1)!] / (p_n!) = [(p_n / 2)!(p_n / 2 -1)!] / (p_n!)$ and Stirling's formula implies that $[(p_n / 2)!(p_n / 2 -1)!] / (p_n!) \sim \sqrt{2\pi/ p_n} \ 2^{-p_n}$. Since $H_{p_n}=\log p_n+O(1)$,
\begin{align*}
\Pi(Z_n = Z_{n}^{*})\leq \exp(-CK_{n}^{*}\log \log (p_n+2))\left( C\sqrt{\frac{2\pi}{p_n}}\frac{1}{2^{p_n}}\right) ^{K_{n}^{*}}
\leq \exp(-CK_{n}^{*}p_n).
\end{align*}
Comparing the results obtained with and without sparsity conditions, we find the lower-bound with sparsity condition is very likely to be larger than the upper-bound without sparsity condition. To consider a very extreme case, when $s_n=\log (p_n+1)$,  the lower-bound 
\noindent$\exp(-CK_{n}^{*}(\log (p_n+1))^2)$ is much larger than $\exp(-CK_{n}^{*}p_n)$. \qed
\end{remark}


We present the main theorem of this paper, which proves that for a sequence of true latent feature matrix $Z_{n}^*$ that satisfy the sparsity condition in Definition \ref{def:sparsity}, the posterior distribution of the similarity matrix $Z_nZ_n^T$ converges to $Z_{n}^* Z_{n}^{*T}$.
The theorem eventually leads to the main theoretical result in Remark \ref{rem-main} later.

\begin{theorem}\label{thm1}

Consider a sequence of sparse binary matrices $\{ Z_{n}^*, n = 1, 2, \ldots \}$ as in Definition \ref{def:sparsity} and the prior in either of the two cases of Lemma \ref{lem1}. 
For each $n$, suppose the observations are generated from $\bx_i^{T}\stackrel{i.i.d.}{\sim} N(0,Z_{n}^* Z_{n}^{*T}+\bm{I}_{p_n})$ for $i=1,2,\cdots,n$. Let 
$$
\epsilon_n=\frac{\max \left\{ \sqrt{p_n},\sqrt{s_nK_{n}^{*}\log (p_n+1)} \right\}}{\sqrt{n}}\max\{1,\|Z_{n}^*Z_{n}^{*T} \|\}.
$$
If $\epsilon_n \rightarrow 0$ as $n\rightarrow \infty$, then we have
\begin{align*}
E_{Z_{n}^{*}} \left[\Pi \left(\|Z_nZ_n^T-Z_{n}^{*}Z_{n}^{*T}\|\leq C\epsilon_n \mid X \right) \right] \rightarrow 1
\end{align*}
as $n\rightarrow\infty$ for some positive constant $C$. In other words, $\epsilon_n$ is the posterior contraction rate under the spectral norm.
\end{theorem}

For the sequence of sparse binary matrices $\{ Z_{n}^{*}, n = 1, 2, \ldots \}$ considered in Theorem \ref{thm1}, if $\|Z_{n}^{*}Z_{n}^{*T}\|\leq M_n$ for some $M_n\geq 1$ and $\tilde{\epsilon}_n \triangleq \frac{\max \left\{ \sqrt{p_n},\sqrt{s_n K_{n}^{*} \log(p_n + 1)} \right\}}{\sqrt{n}}M_n\rightarrow 0$ as $n\rightarrow  \infty $, 
then
$\tilde{\epsilon}_n$ is a valid posterior contraction rate. 
We show in the following corollary that if the number of features possessed by each object is upper bounded, then a new posterior contraction rate with a simpler expression (compared to Theorem \ref{thm1}) can be derived.

\begin{corollary}\label{cor1}
Consider a sequence of sparse binary matrices $\{ Z_{n}^*, n = 1, 2, \ldots \}$ as in Definition \ref{def:sparsity}
where $Z_{n}^*=(z_{jk}^*)_{p_n \times K_{n}^*}$. Suppose that there exists $q_n\geq 1$ such that
$$
\sup\limits_{1\leq j\leq p_n}\Big(\sum\limits_{k=1}^{K_{n}^{*}}z_{jk}^{*}\Big) \leq q_n,
$$ 
i.e., the number of features possessed by each object (non-zero entries of each row of $Z_{n}^{*}$) is upper bounded by $q_n$. Given the same assumptions in Theorem \ref{thm1} except for replacing 
$\epsilon_n\rightarrow 0$ by 
\begin{align*}
\tilde{\epsilon}_n\triangleq \frac{\max \left\{ \sqrt{p_n},\sqrt{s_nK_{n}^{*}\log(p_n+1)} \right\} }{\sqrt{n}}s_nq_n\rightarrow 0
\end{align*}
as $n\rightarrow \infty$, we have
\begin{align*}
E_{Z_{n}^{*}} \left[\Pi \left( \|Z_n Z_n^T - Z_{n}^{*} Z_{n}^{*T}\|\leq C\tilde{\epsilon}_n \mid X \right) \right]\rightarrow 1
\end{align*}
as $n\rightarrow\infty$ for some positive constant $C$. Specifically,
\begin{enumerate}
\item if there is no additional condition on $q_n$,
\begin{align*}
\tilde{\epsilon}_n=\frac{\max \left\{\sqrt{p_n},\sqrt{s_nK_{n}^{*}\log(p_n+1)} \right\}}{\sqrt{n}}s_nK_{n}^{*};
\end{align*}

\item if $q_n$ is a constant, 
\begin{align*}
\tilde{\epsilon}_n=\frac{\max \left\{\sqrt{p_n},\sqrt{s_nK_{n}^{*}\log(p_n+1)} \right\}}{\sqrt{n}}s_n.
\end{align*}
\end{enumerate}
\end{corollary}

\begin{remark}\label{rem-main}
Consider the second case of Corollary \ref{cor1}. If (1)
$s_nK_{n}^{*}\log(p_n+1)=O(p_n)$ and (2) $s_n=O(\log(p_n+1))$,
then $\tilde{\epsilon}_n\leq
  C\frac{\sqrt{p_n}\log(p_n+1)}{\sqrt{n}}$;  therefore, if $p_n(\log(p_n+1))^2=o(n)$,  then $ \frac{\sqrt{p_n}\log(p_n+1)}{\sqrt{n}}$ is a valid posterior contraction rate.
In other words, to ensure posterior convergence, we only need $n$ to increase a little bit faster than $p_n$, given the assumptions (1), (2) and the condition in the second case of Corollary \ref{cor1}.
\end{remark}


\section{Posterior Inference Based on MCMC}
\label{sec:MCMC}

We have specified the hierarchical models including the sampling model $p(X \mid Z)$ and the prior models for the parameters $\Pi(Z \mid \alpha)$ and $p(\alpha)$. In particular, $p(X \mid Z)$ is the latent feature model (Equation \ref{eq:sampling_model}), $\Pi(Z \mid \alpha)$ is the IBP or pIBP prior (Equation \ref{eq:pmf-pIBP}) and $p(\alpha) = \text{Gamma}(1, 1)$. For the theoretical results in Section \ref{sec:rate}, we integrate out $\alpha$. For posterior inference, we keep $\alpha$ so that the conditional distributions can be obtained in closed form.

We use Markov chain Monte Carlo simulations to generate samples from the posterior $(Z, \alpha) \sim \Pi(Z, \alpha \mid X) \propto p(X \mid Z) \Pi(Z\mid \alpha) p(\alpha)$. After iterating sufficiently many steps, the samples of $Z$ drawn from the Markov chain approximately follow $\Pi(Z \mid X)$.
Gibbs sampling transition probabilities can be used to update $Z$ and $\alpha$, as described in \cite{griffiths2011indian} and \cite{miller2012phylogenetic}.

To overcome trapping of the Markov chain in local modes in the high-dimensional setting, we use parallel tempering Markov chain Monte Carlo (PTMCMC) \citep{geyer1991markov} in which several Markov chains at different temperatures run in parallel and interchange the states across each other. 
In particular, the target distribution of the Markov chain indexed by temperature $T$ is
\begin{align*}
\Pi^{(T)}(Z, \alpha \mid X) \propto p(X \mid Z)^{\frac{1}{T}} \Pi(Z\mid \alpha) p(\alpha).
\end{align*}
Parallel tempering helps the original Markov chain (the Markov chain whose temperature $T$ is 1) avoid getting stuck in local modes and approximate the target distribution efficiently.
 
We give an algorithm below for sampling from $\Pi(Z,\alpha \mid X)$ where $\Pi(Z \mid \alpha)$ follows IBP. 
The algorithm describes in detail how PTMCMC can be combined with the Gibbs sampler in \cite{griffiths2011indian}.
The algorithm iterates Step 1 and Step 2 in turn.


\noindent\textbf{Step 1 (Updating $Z$ and $\alpha$).}
Denote by $z_{jk}$ the entries of 
$Z$, $j = 1, \ldots, p$ and $k = 1, \ldots, K$. 
We update $Z$ by row.
For each row $j$, we iterate through the columns $k = 1, \ldots, K$.
We first make a decision to drop the $k$-th column of $Z$, $\bm{z}_k$, if and only if $m_{(-j)k} = \sum_{j'\neq j}z_{j'k} = 0$. 
In other words, if feature $k$ is not possessed by any object other than $j$, then the $k$-th column of $Z$ should be dropped, regardless of whether $z_{jk}=0$ or $1$.

If the $k$-th column is not dropped, we sample $z_{jk}$ from
\begin{align*}
\Pi^{(T)}(z_{jk} \mid \ldots)\propto p(X \mid Z)^{\frac{1}{T}}\Pi(z_{jk} \mid Z_{(-jk)},\alpha),
\end{align*}
where $Z_{(-jk)}$ represents all entries of $Z$ except $z_{jk}$ and $P(X \mid Z)$ is determined by $\bx_i^T \mid Z \stackrel{i.i.d.}{\sim}N(0,ZZ^T+\bm{I}_p)$, in which $\bx_i$'s are rows of $X$.
The conditional prior $\Pi(z_{jk} \mid Z_{(-jk)},\alpha)$ only depends on $\bm{z}_{(-j)k}$ (i.e., the $k$-th column of $Z$ excluding $z_{jk}$). Specifically, 
\begin{align*}
\Pi(z_{jk}=1 \mid Z_{(-jk)},\alpha)=\frac{m_{(-j)k}}{p}.
\end{align*}

After updating all entries in the $j$-th row, we add a random number of $K_{j}^{+}$ columns (features) to $Z$. The $K_{j}^{+}$ new features are only possessed by object $j$, i.e. only the $j$-th entry is 1 while all other entries are 0.
Let $Z_{+}$ denote the feature matrix after $K_{j}^{+}$ columns are added to the old feature matrix.
The conditional posterior distribution of $K_{j}^{+}$ is
\begin{align*}
P^{(T)}(K_{j}^{+} \mid \ldots) \propto p(X \mid Z_{+})^{\frac{1}{T}}P(K_{j}^{+} \mid \alpha),
\end{align*}
in which $P(K_{j}^{+} \mid \alpha)$ is the prior distribution of $K_{j}^{+}$ under IBP, 
$P(K_{j}^{+} \mid \alpha) = \text{Pois}\left(\alpha/p\right)$.
The support of $P^{(T)}(K_{j}^{+} \mid \ldots)$ is $\mathbb{N}$. For easier evaluation of $P^{(T)}(K_{j}^{+} \mid \ldots)$, we work with an approximation by truncating $P^{(T)}(K_{j}^{+} \mid \ldots)$ at level $K_{\max}^+$, similar to the idea of truncating a stick-breaking prior in \cite{ishwaran2001gibbs}. The value $K_{\max}^{+}$ is the maximum number of new columns (features) that can be added to $Z$ each time we update the $j$-th row.
Denote by $\tilde{P}^{(T)}(K_{j}^{+} \mid \ldots)$ the truncated conditional posterior, we have
\begin{align}
\tilde{P}^{(T)}(K_{j}^{+} = k \mid \ldots)=\frac{P^{(T)}(k \mid \ldots)}{\sum_{k' = 0}^{K_{\max}^+} P^{(T)}(k' \mid \ldots) },\quad \text{for}\ k =0,1,\cdots,K_{\max}^{+}. \label{ap1}
\end{align}

Lastly, we update $\alpha$.
Given $Z$, the observed data $X$ and $\alpha$ are conditionally independent, which implies that the conditional posterior distribution of $\alpha$ at any temperature $T$ is the same. We sample $\alpha$ from
\begin{align*}
(\alpha \mid \ldots) \sim \text{Gamma}
\Big(
K+1,\Big(\sum\limits_{j=1}^{p}\frac{1}{j}+1\Big)^{-1} 
\Big).
\end{align*}

\noindent\textbf{Step 2 (Interchanging States across Parallel Chains).}
We sort the Markov chains in descending order by their temperatures $T$. The next step of PTMCMC is interchanging states between adjacent Markov chains. Let $\left( Z^{(T)}, \alpha^{(T)} \right)$ denote the state of the Markov chain indexed by temperature $T$.
Suppose we run $N$ parallel chains with descending temperatures $T_N, T_{N-1}, \ldots, T_2, T_1$, where $T_1 = 1$. Sequentially for each $i = N, N-1, \ldots, 2$, we propose an interchange of states between $\left( Z^{(T_i)},\alpha^{(T_i)} \right)$ and $\left( Z^{(T_{i-1})}, \alpha^{(T_{i-1})} \right)$ and accept the proposal with probability
\begin{align*}
A_{T_i, T_{i-1}}=\min\Big\{ \Big( \frac{p(X \mid Z^{(T_{i-1})})}{p(X \mid Z^{(T_i)})}\Big)^{\frac{1}{T_i}-\frac{1}{T_{i-1}}},1\Big\}.
\end{align*}

\section{Examples}
\label{sec:exap}
\subsection{Simulation Studies}
\label{sec:simulation}

We conduct simulation to examine the convergence of the posterior distribution of $Z_n Z_n^T$ (under IBP mixture prior) to the true similarity matrix $Z_n^* Z_n^{*T}$. 
We consider several true similarity matrices with different sample
sizes $n$'s and explore the contraction of the posterior
distribution. Hereinafter, we suppress the index $n$.

\paragraph{Simulations under the Sparsity Condition of $Z^*$}\mbox{}

In this simulation, the true latent feature matrix, $Z^*=(z_{jk}^*)_{p\times K^*}$, is randomly generated under the sparsity condition in Definition \ref{def:sparsity} (i.e., $m_k\leq s$ for $k = 1, \ldots, K^*$, where $s$ is relatively small compared with $p$).
We set $s=10$, $K^*= 10$ and $p=50, 100, $ or $150$. We use $s/p$ to measure the sparsity (as in Definition \ref{def:sparsity}). 

Once $Z^*$ is generated, the samples $\bx_1, \bx_2, \cdots, \bx_n$ are generated under the sampling model (Equation \ref{eq:sampling_model}) with
\begin{align*}
\bx_i^T \stackrel{i.i.d.}{\sim} N(0, Z^* Z^{*T}+\bm{I}_p),
\end{align*}
where $\bx_i$'s are rows of $X$. For each value of $p$, we conduct 4 simulations, with sample sizes $n = 20, 50, 80$ or $100$.

Given simulated data $X$, we use the proposed PTMCMC algorithm introduced in Section \ref{sec:MCMC} to sample from the posterior distribution $\Pi(Z \mid X)$, setting $K_{\max}^{+}$ in (\ref{ap1}) at $10$. 
We set the number of parallel Markov chains $N = 11$ with geometrically spaced temperatures. Namely, the ratio between adjacent temperatures $\beta \triangleq T_{i+1} / T_i =1.2$, with $T_1=1$. We run $1,000$ MCMC iterations.
The chains converge quickly and mix well based on basic MCMC diagnostics.

We repeat the simulation scheme 40 times, each time generating a new data set with a new random seed and applying the PTMCMC algorithm for inference. We report two summaries in Table \ref{tab1} based on the 40 simulation studies. The entries are the true values of $(n, p)$ (first column), $p/n$ (second column), average (across the 40 simulation studies) of the $1,000$-th MCMC value of $K$ (third column), average residual $\|ZZ^T-Z^*Z^{*T}\|$ (across the 40 simulation studies) based on the $1,000$-th MCMC value of $Z$ (fourth column).

For fixed $p$, $ZZ^T$ converges to $Z^*Z^{*T}$ as sample size increases. When $n=100$, $\|ZZ^T-Z^*Z^{*T}\|$ is very close to $0$ and $K$ is close to the truth  $K^* = 10$.
On the other hand, for fixed $n$, increasing $p$ does not make the posterior distribution of $ZZ^T$ significantly less concentrated at the true $Z^*Z^{*T}$, implying that the inference is robust to high-dimensional matrix of $Z$ as long as the true matrix $Z^*$ is sparse. This verifies the theoretical results we have reported early on. 

\begin{table}[h!]
\begin{center}
	\begin{tabular}{|r|c|cc|}
		\hline
		$(n,p)$& $p/n$& $K_{sd}$& $\|ZZ^T-Z^*Z^{*T}\|_{sd}$ \\
		\hline
		$(20,50)$&$2.5$&$13.550_{2.385}$&$16.349_{5.315}$\\
		$(50,50)$&$1$&$10.700_{0.791}$&$3.444_{2.205}$\\
		$(80,50)$&$0.625$&$10.300_{0.564}$&$0.591_{1.117}$\\
		$(100,50)$&$0.5$&$10.275_{0.506}$&$0.478_{0.951}$\\ \hline
		$(20,100)$&$5$&$12.650_{1.718}$&$12.505_{3.029}$\\
		$(50,100)$&$2$&$10.750_{0.776}$&$3.495_{1.552}$\\
		$(80,100)$&$1.25$&$10.275_{0.506}$&$0.709_{1.315}$\\
		$(100,100)$&$1$&$10.250_{0.494}$&$0.343_{0.734}$\\ \hline
		$(20,150)$&$7.5$&$12.600_{2.073}$&$11.339_{2.618}$\\
		$(50,150)$&$3$&$10.675_{0.797}$&$3.387_{1.876}$\\
		$(80,150)$&$1.875$&$10.282_{0.456}$&$0.472_{0.884}$\\
		$(100,150)$&$1.5$&$10.250_{0.439}$&$0.250_{0.439}$\\		
		\hline
	\end{tabular}
	\caption{Simulation results based on the sample value of $K$ and $Z$ from the $1,000$-th MCMC iteration. Entries in normal size and in the subscript (columns 3 and 4) are the averages and standard deviations of 40 independent simulation studies.}
	\label{tab1}
\end{center}
\end{table}

\paragraph{Sensitivity of Sparsity}\mbox{}

In this part, we use simulation results to demonstrate the effect of sparsity of $Z^*$ in the posterior convergence. We set $p=50$, $K^* = 10$  and $n = 100$, and use different $s=10,15,18,20, 23, 25$, reducing the sparsity of $Z^*$ gradually.   We generate $Z^*$ and $X$ the same way as in the previous simulation. We set the number of parallel Markov chains $N=11$ and $K_{\max}^{+}=10$. To increase the frequency of interchange between adjacent chains, we reduce $\beta$ to 1.15.


\begin{table}[h!]
\begin{center}
\begin{tabular}{|c|c|cc|}
\hline
$s$ & $s/p$& $K_{sd}$& $\|ZZ^T-Z^*Z^{*T}\|_{sd}$ \\
\hline
10 & 0.2   & $10.225_{0.480}$ & $0.200_{0.405}$\\
15 & 0.3   & $10.100_{0.304}$ &  $0.100_{0.304}$\\
18 & 0.36 & $10.225_{0.423}$ & $1.212_{5.454}$ \\
20 & 0.4   & $10.275_{0.554}$ & $1.959_{5.146}$ \\ 
23 & 0.46 & $10.850_{0.802}$ & $20.192_{26.867}$\\
25 & 0.5   & $11.150_{0.736}$  & $24.552_{19.141}$\\
		\hline
	\end{tabular}
\end{center}
\caption{The true$(n,p,K^*)=(100,50,10)$. Simulation results based on the sample value of $K$ and $Z$ from the $1,000$-th MCMC iteration. Entries in normal size and in the subscript (rows 3 and 4) are the average and standard deviations of 40 independent simulation studies.} 
\label{tab:sens}
\end{table}

Table \ref{tab:sens} reports the simulation results.
As $Z^*$ becomes less sparse, the posterior distribution of $Z$ becomes less concentrated on $Z^*$, in terms of both $K$ and $\|ZZ^T-Z^*Z^{*T}\|$. Specifically, for $s\in \{15,20,25\}$, when $s$ increases by 5, $\|ZZ^T-Z^*Z^{*T}\|$ inflates more than 10 fold.

\subsection{Application for Proteomics}
\label{sec:realdata}

We apply the binary latent feature model to the analysis of a reverse phase protein arrays (RPPA) dataset from The Cancer Genome Atlas (TCGA, \url{https://tcga-data. nci. nih. gov/tcga/}) downloaded by TCGA Assembler 2 \citep{wei2017tcga}. 
The RPPA dataset records the levels of protein expression based on incubating a matrix of biological samples on a microarray with specific antibodies that target corresponding proteins  \citep{sheehan2005use, spurrier2008reverse}. 
 We focus on patients categorized as 5 different cancer types, including breast cancer (BRCA), diffuse large B-cell lymphoma (DLBC), glioblastoma multiforme (GBM), clear cell kidney carcinoma (KIRC) and lung adenocarcinoma (LUAD). 
Data of $n = 157$ proteins are available. 
We randomly choose $p = 100$ patients for our analysis, with an equal number of 20 patients for each cancer type. Note that we consider proteins as experimental units and patients as objects in the data matrix with an aim to allocate latent features to patients (not proteins). This will be clear later when we report the inference results.
The rows of the data matrix $X_{n \times p}$ are standardized, with $\sum_{j = 1}^p x_{ij} = 0$ for $i = 1, \ldots, n$.
Since the patients are from different cancer types, we expect the data to be highly heterogeneous and the number of common features that two patients share to be small, i.e the feature matrix is sparse. 

\begin{figure}[h!]
\center
\includegraphics[height = 8.5 cm]{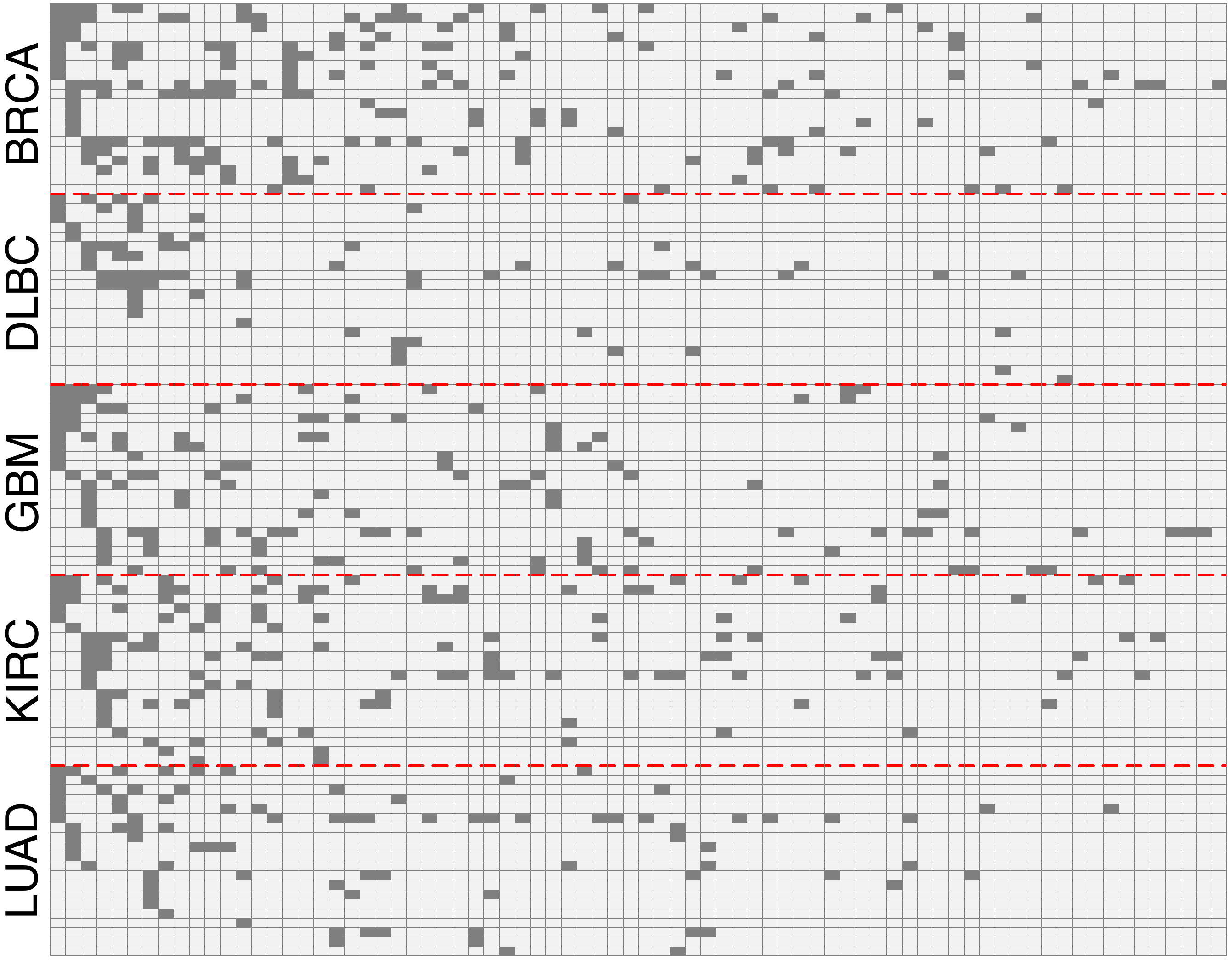}
\caption{The inferred binary feature matrix $\hat{Z}$ for the TCGA RPPA dataset. The dataset consists of 100 patients, with 20 patients for each of the 5 cancer type, BRCA, DLBC, GBM, KIRC and LUAD. A shaded gray rectangle indicates the corresponding patient $j$ possesses feature $k$, i.e. the corresponding matrix element $\hat{Z}_{jk} = 1$. The columns are in descending order of the number of objects possessing each feature. The rows are reordered for better display. 
}
\label{fig:real_data_Z_hat}
\end{figure}

We fit the binary latent feature model under the IBP prior 
$\Pi(Z  \mid  \alpha)$ and $\alpha \sim\text{Gamma}(1,1)$. In addition, rather than fixing $\sigma^2 = \sigma_a^2 = 1$, we now assume inverse-Gamma distribution priors on them, $\sigma^2 \sim \text{IG}(1, 1)$ and $\sigma_a^2 \sim \text{IG}(1, 1)$. We run 1,000 MCMC iterations, as in the simulation studies. We also repeat the MCMC algorithm 3 times with different random seeds and do not observe substantial difference across the runs. We report point estimates $(\hat{Z}, \hat{\alpha}, \hat{\sigma}^2, \hat{\sigma}_a^2) $ according to the maximum a posteriori (MAP) estimate,
\begin{align*}
(\hat{Z}, \hat{\alpha}, \hat{\sigma}^2, \hat{\sigma}_a^2) = \argmax_{(Z, \alpha, \sigma^2, \sigma_a^2)} p(Z, \alpha, \sigma^2, \sigma_a^2 \mid X).
\end{align*}

Figure \ref{fig:real_data_Z_hat} shows the inferred binary feature matrix $\hat{Z}$, where a shaded gray rectangle indicates the corresponding patient $j$ possesses feature $k$. 
For the real data analysis we do not 
know the true latent feature matrix and its sparsity, but can use the estimated $\hat{Z}$ to approximate the truth.
From the $\hat{Z}$ matrix, we find that a feature is possessed by at most 
  $\hat{s} = 31$  patients, and therefore the sparsity of $\hat{Z}$ is about  $\hat{s} / p = 0.31$. 
If the estimated sparsity is close to the truth, according to simulation results (Section \ref{sec:simulation} Table \ref{tab:sens}), the posterior distribution of $Z$ should be highly concentrated at the truth $Z^*$.

\paragraph{Biological Interpretation of the Features}\mbox{}

We report the unique genes for the top 10 proteins that have the largest loading values $\hat{a}_{ki}$ for the five most popular features. That is, for the top five features $k$ possessed by the largest numbers of patients (the first five columns in Figure \ref{fig:real_data_Z_hat}), we report the proteins with the largest $\hat{a}_{ki}$ values.  The $\hat{a}_{ki}$ values are posterior mean from the MCMC samples, in which parameters $a_{ki}$'s are sampled from their full conditional distributions. This additional sampling step for $a_{ki}$'s is added to the proposed PTMCMC algorithm for the purpose of assessing the biological implication of the features. It is a simple Gibbs step as the full conditional distributions of $a_{ki}$'s are known Gaussian distributions. Table \ref{table:app:gene-feature} in the supplement lists these genes and their feature membership. 
We conduct the gene set enrichment analysis (GSEA) in \cite{subramanian2005gene} comparing the genes in each feature with pathways from the Kyoto Encyclopedia of Genes and Genomes (KEGG) database \citep{kanehisa2000kegg}. The GSEA analysis reports back the enriched pathways and the corresponding genes, which is listed in Table \ref{table:app:gene-pathway} in the supplement. We observe the following findings. 

\textbf{Feature 1.} Among all the features, only genes in feature 1 are enriched in  the cell cycle pathway. They are also enriched in the p53 signaling pathway. This indicates that feature 1 might be related to cell death and cell cycle.  
 \textbf{Feature 2.} Feature 2 is enriched in many different types of pathways, which may be caused by its two key gene members: NRAS:4893 (oncogene) and PTEN:5728 (tumor suppressor). These two genes play key roles in the PI3K-Akt signaling pathway and also regulate many other cancer-related pathways.
 \textbf{Feature 3.} Genes in feature 3 are enriched in inflammation related pathways, such as the non-alcoholic fatty liver disease, Hepatitis B, viral carcinogenesis, hematopoietic cell lineage and phagosome pathways. This means feature 3 is mostly related to inflammation. 
\textbf{Feature 4 and 5.}
Genes in features 4 and 5 are enriched in the largest number of pathways that are similar, with the exception of the p53 pathway (enriched with feature 4 but not 5) and the Jak-STAT signaling pathway (enriched with feature 5 but not 4). This indicates that feature 4 is more related to intracellular factors like DNA damage, oxidative stress and activated oncogenes, while feature 5 is more related to extracellular factors such as cytokines and growth factors.

Depending on the possession of the first five features, the patients
in each cancer type can be further divided into potential
molecular subgroups. For example, most BRCA patients possess
features 1, 2 or 3, which indicate that these tumors are related to
cell death and cell cycle (features 1 \& 2), or
inflammation (feature 3) pathways; most of the
GBM and KIRC patients possess features 1, 2, 3 or 4, indicating an
additional subgroup of patients with tumor associated to DNA damage
(feature 4). The DLBC patients are highly heterogeneous, as
many of them do not possess any of the first five main
features. This has been well recognized in the
literature \citep{zhang2013genetic}. Lastly,
LUAD seems to have two subgroups, possessing mostly feature 1 or 2,
respectively. They correspond to cell death and cell cycle functions,
which suggest that these two subgroups of cancer could be related to
abnormal cell death and cell cycle regulation.

We also note that there are other informative features besides the five mentioned above.  For example, feature 16 is only possessed by BRCA patients, in which the top genes include ESR1, AR, GATA3, AKT1, CASP7, ETS1, BCL2, FASN and CCNE2, all of which have been shown closely related to breast cancer (see \cite{clatot2017esr1, cochrane2014role, takaku2015gata3, ju2007akt1, chaudhary2016overexpression, furlan2014ets, dawson2010bcl2, menendez2017fatty,tormo2017role}).

\section{Conclusion and Discussion}
\label{sec:discussion}

Our main contributions in this paper are (1) reducing the requirement on the growth rate of sample size with respect to dimensionality that ensures posterior convergence of IBP mixture or pIBP mixture under proper sparsity condition, (2) proposing an efficient MCMC scheme for sampling from the model, and (3) demonstrating the practical utility of the derived properties through an analysis of an RPPA dataset.  The sparsity condition is mild and interpretable, making real-case applications possible.
This result guarantees the validity of using IBP mixture or pIBP mixture for posterior inference in high dimensional settings theoretically.

There are several directions along which we plan to investigate further.
First, since the assumptions made on the true latent feature matrix $Z_n^*$ are quite mild,  the posterior convergence in Theorem \ref{thm1} only holds when $p_n = o(n)$.
It is of interest whether posterior convergence still holds when $p_n$ increases faster, \eg, $p_n \gg n$. 
As a trade-off, results with a faster-increasing $p_n$ would likely require additional assumptions on $Z_n^*$, such as the Assumption 3.2 (A3) in \cite{pati2014posterior}.
It is also of interest to explore whether the contraction rate in Theorem \ref{thm1} can be further improved with additional assumptions. This is closely related to the problem of minimax rate optimal estimators for $Z_n^*{Z_n^*}^T$, or more broadly, the covariance matrix of random samples, which has been partially addressed in \cite{pati2014posterior}.

Another potential direction for further investigation is to extend the
latent feature model (\ref{eq:lfm}) to a more general latent factor
model, in which the binary matrix $Z$ is replaced with a real-valued
factor matrix $G$. The binary matrix $Z$ is then used to indicate the
sparsity of $G$. See, \eg, \cite{knowles2011nonparametric}. To prove
posterior convergence for such a model, 
Lemma \ref{lem1} needs to be modified 
based on the factor loading matrix, such as Lemma 9.1 in \cite{pati2014posterior}.

Throughout this paper we measure the difference between the similarity matrices by the spectral norm. 
Other matrix norms, such as the Frobenius norm, may be explored.
Our current results focus on the posterior convergence of $Z_n Z_n^T$ rather than $Z_n$ itself due to the identifiability issue of $Z_n$ arose from (\ref{eq:sampling_model}). A future direction is to investigate to what extent can $Z_n^*$ be estimated, and a Hamming distance like measure between the feature matrices can be considered.

Finally, we are working on general hierarchical models that embed sparsity into the model construction.

\bibliographystyle{apalike}
\bibliography{ref_IBP}

\clearpage

\section*{Supplementary Materials}

\beginsupplement

\subsection{Derivation of the PMF of pIBP}
\label{app:deriv_pibp}

In Step 1 of the construction of pIBP, recall
\begin{align*}
P_{\tilde{K}}(\tilde{Z}  \mid  \alpha, \tree) = \prod \limits_{k=1}^{\tilde{K}} \int P(\tilde{\bm z}_k \mid \pi_k, \tree) p(\pi_k \mid \alpha) d \pi_k.
\end{align*}
Here $P(\tilde{\bm z}_k \mid \pi_k, \tree)$ is defined by the universal tree structure $\tree$ for any $k = 1, \ldots, \tilde{K}$ according to \cite{miller2012phylogenetic}.

Denote by $\tilde{\bm{z}}_k = (z_{1k},z_{2k},\cdots,z_{pk})^T$. For each $k$, we treat the tree as a directed graph with random variables at the interior nodes and $z_{jk}$'s at the leaf nodes, $j = 1, \ldots, p$. See Figure \ref{fig:eg_tree2}. Given $\pi_k$, for any variable $x$ at the parent node (including the root node) and any variable $y$ at the child node (including the leaf node, in which case $y = z_{jk}$), if the edge between $x$ and $y$ has length $t$, then
\begin{align}
\begin{split}
&P(y=0 \mid x=0,\pi_k) = \exp(-\gamma_kt), \\
&P(y=1 \mid x=0,\pi_k) = 1 - \exp(-\gamma_kt), \\
&P(y=0 \mid x=1,\pi_k) = 0, \\
&P(y=1 \mid x=1,\pi_k) = 1, 
\end{split}
\label{eq:pIBP_cond_probs}
\end{align}
where $\gamma_k= -\log(1-\pi_k)$. The value of the random variable at the root node is always fixed at 0. The joint distribution of all random variables on the tree (nodes and leaves) is uniquely defined by the conditional probabilities (see Figure \ref{fig:eg_tree2}). So is $P(\tilde{\bm z}_k \mid \pi_k, \tree)$. 
Note that given $\bpi$ and $\tree$, the columns of $\tilde{Z}$ are conditionally independent, i.e. $P_{\tilde{K}}(\tilde{Z} \mid \alpha, \tree)$ is column-exchangeable. If $\tilde{Z}_1$ is generated by permuting columns of $\tilde{Z}$, then $P_{\tilde{K}}(\tilde{Z} \mid \alpha, \tree)=P_{\tilde{K}}(\tilde{Z}_1 \mid \alpha, \tree)$.

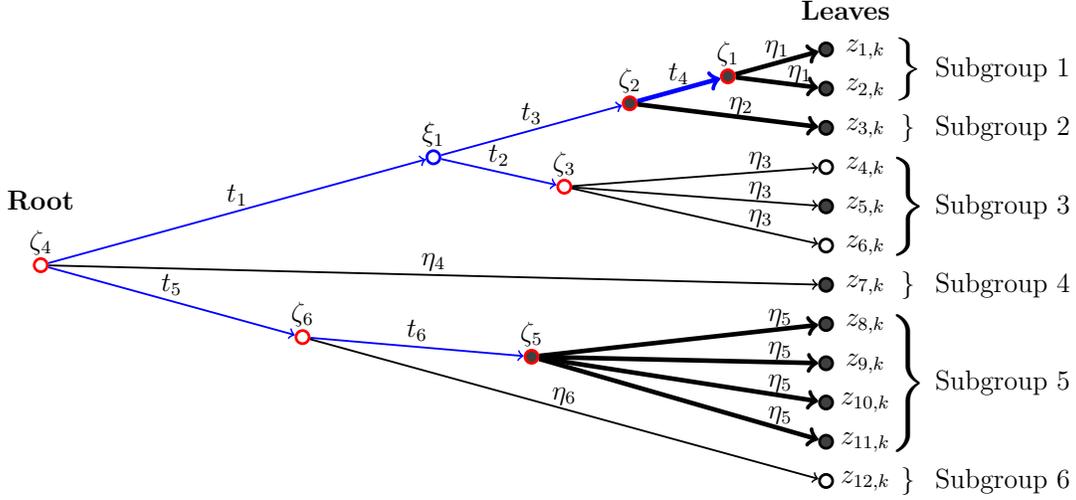
\begin{figure}[h!]
\begin{center}
\scalebox{0.87}{
\begin{tikzpicture}
[every node/.style = {very thick}]
    \node[root, draw = red] at (0, 0)         (root) {};
    
    \node[dummy, draw = blue] at (6, 1.65)   (xi1) {};
    \node[dummy, draw = red, fill = deep-gray] at (10.5, 2.8875)   (zeta1) {};
    \node[dummy, draw = red, fill = deep-gray] at (9, 2.475)   (zeta2) {};
    \node[dummy, draw = red] at (8, 1.2)   (zeta3) {};
    \node[dummy, draw = red, fill = deep-gray] at (7.5, -1.4)   (zeta5) {};
    \node[dummy, draw = red] at (4, -1.1)   (zeta6) {};
    
	\node[dummy, draw = black, fill = deep-gray] at (12, 3.3)   (z1) {};
    \node[dummy, draw = black, fill = deep-gray] at (12, 2.7)   (z2) {};
    \node[dummy, draw = black, fill = deep-gray] at (12, 2.1)   (z3) {};
    \node[dummy, draw = black] at (12, 1.5)   (z4) {};
    \node[dummy, draw = black, fill = deep-gray] at (12, 0.9)   (z5) {};
    \node[dummy, draw = black] at (12, 0.3)   (z6) {};
    \node[dummy, draw = black, fill = deep-gray] at (12, -0.3)   (z7) {};
    \node[dummy, draw = black, fill = deep-gray] at (12, -0.9)   (z8) {};
    \node[dummy, draw = black, fill = deep-gray] at (12, -1.5)   (z9) {};
    \node[dummy, draw = black, fill = deep-gray] at (12, -2.1)   (z10) {};
    \node[dummy, draw = black, fill = deep-gray] at (12, -2.7)   (z11) {};
    \node[dummy, draw = black] at (12, -3.3)   (z12) {};    
    
    \node[env] at (0, 1)    (root_text) {\textbf{Root}};
    \node[env] at (10.5, 3.2375)    (zeta1_rv) {$\zeta_1$};
    \node[env] at (9, 2.825)    (zeta2_rv) {$\zeta_2$};
    \node[env] at (8, 1.55)    (zeta3_rv) {$\zeta_3$};
    \node[env] at (0, 0.35)    (root_rv) {$\zeta_4$};
    \node[env] at (7.5, -1.05) (zeta5_rv) {$\zeta_5$};
    \node[env] at (4, -0.75) (zeta6_rv) {$\zeta_6$};
    \node[env] at (6, 2) (xi1_rv) {$\xi_1$};
    
    \node[env] at (12.3, 3.9) (leaves_text) {\textbf{Leaves}};
    \node[env] at (12.6, 3.3) (z1_rv) {$z_{1, k}$};
    \node[env] at (12.6, 2.7) (z2_rv) {$z_{2, k}$};
    \node[env] at (12.6, 2.1) (z3_rv) {$z_{3, k}$};
    \node[env] at (12.6, 1.5) (z4_rv) {$z_{4, k}$};
    \node[env] at (12.6, 0.9) (z5_rv) {$z_{5, k}$};
    \node[env] at (12.6, 0.3) (z6_rv) {$z_{6, k}$};
    \node[env] at (12.6, -0.3) (z7_rv) {$z_{7, k}$};
    \node[env] at (12.6, -0.9) (z8_rv) {$z_{8, k}$};
    \node[env] at (12.6, -1.5) (z9_rv) {$z_{9, k}$};
    \node[env] at (12.6, -2.1) (z10_rv) {$z_{10, k}$};
    \node[env] at (12.6, -2.7) (z11_rv) {$z_{11, k}$};
    \node[env] at (12.6, -3.3) (z12_rv) {$z_{12, k}$};
    
    \node[env] at (3, 1.075) (edge1) {$t_1$};
    \node[env] at (7, 1.675) (edge2) {$t_2$};
    \node[env] at (7.5, 2.3125) (edge3) {$t_3$};
    \node[env] at (9.75, 2.93125) (edge4) {$t_4$};
    \node[env] at (2, -0.3) (edge5) {$t_5$};
    \node[env] at (5.75, -1) (edge6) {$t_6$};
    
    \node[env] at (11.25, 3.3) (eta1_1) {$\eta_1$};
    \node[env] at (11.6, 2.93) (eta1_2) {$\eta_1$};
    \node[env] at (10.7, 2.4375) (eta2) {$\eta_2$};
    
    \node[env] at (11, 1.6) (eta3_1) {$\eta_3$};
    \node[env] at (11, 1.15) (eta3_2) {$\eta_3$};
    \node[env] at (11, 0.72) (eta3_3) {$\eta_3$};
    
    \node[env] at (6, 0.05) (eta4) {$\eta_4$};
    
    \node[env] at (11.3, -0.8) (eta5_1) {$\eta_5$};
    \node[env] at (11.3, -1.3) (eta5_2) {$\eta_5$};
    \node[env] at (11.3, -1.8) (eta5_3) {$\eta_5$};
    \node[env] at (11.3, -2.3) (eta5_4) {$\eta_5$};
    
    \node[env] at (8, -2) (eta6) {$\eta_6$};

    \draw[draw = blue, ->, line width = 0.8pt] (root) -- (xi1);
    \draw[draw = blue, ->, line width = 0.8pt] (xi1) -- (zeta2);
    \draw[draw = blue, ->, line width = 2pt] (zeta2) -- (zeta1);
    \draw[->, line width = 2pt] (zeta1) -- (z1);
    \draw[->, line width = 2pt] (zeta1) -- (z2);
    \draw[->, line width = 2pt] (zeta2) -- (z3);
    \draw[draw = blue, ->, line width = 0.8pt] (xi1) -- (zeta3);
    \draw[->, line width = 0.8pt] (zeta3) -- (z4);
    \draw[->, line width = 0.8pt] (zeta3) -- (z5);
    \draw[->, line width = 0.8pt] (zeta3) -- (z6);
    \draw[->, line width = 0.8pt] (root) -- (z7);
    \draw[->, line width = 2pt] (zeta5) -- (z8);
    \draw[->, line width = 2pt] (zeta5) -- (z9);
    \draw[->, line width = 2pt] (zeta5) -- (z10);
    \draw[->, line width = 2pt] (zeta5) -- (z11);
    \draw[draw = blue, ->, line width = 0.8pt] (zeta6) -- (zeta5);
    \draw[draw = blue, ->, line width = 0.8pt] (root) -- (zeta6);
    \draw[->, line width = 0.8pt] (zeta6) -- (z12);
    
    \node[env] at (14.7, 3) (subgroup1) {Subgroup 1};
    \node[env, font=\fontsize{23}{0}\selectfont] at (13.25, 3) (subgroup1_1) {{$\big{\}}$}};
    
    \node[env] at (14.7, 2.1) (subgroup2) {Subgroup 2};
    \node[env, font=\fontsize{13}{0}\selectfont] at (13.25, 2.1) (subgroup2_1) {{$\}$}};
    
    \node[env] at (14.7, 0.9) (subgroup3) {Subgroup 3};
    \node[env, font=\fontsize{24}{0}\selectfont] at (13.25, 0.9) (subgroup3_1) {{$\Big{\}}$}};
    
    \node[env] at (14.7, -0.3) (subgroup4) {Subgroup 4};
    \node[env, font=\fontsize{13}{0}\selectfont] at (13.25, -0.3) (subgroup4_1) {{$\}$}};
    
    \node[env] at (14.7, -1.8) (subgroup5) {Subgroup 5};
    \node[env, font=\fontsize{20}{0}\selectfont] at (13.25, -1.8) (subgroup5_1) {{$\Bigg{\}}$}};
    
    \node[env] at (14.7, -3.3) (subgroup6) {Subgroup 6};
    \node[env, font=\fontsize{13}{0}\selectfont] at (13.25, -3.3) (subgroup6_1) {{$\}$}};
    
\end{tikzpicture}
}
\end{center}
\caption{An example of a tree structure $\tree$, which is a directed graph with random variables ($\xi$'s, $\zeta$'s and $z$'s) at the nodes. 
In particular, entries of the $k$-th column of $Z$, $z_{jk}$'s, are at the leaves.
The length of the edge between $\zeta_4$ and $\xi_1$ is $t_1$, and similarly for $t_2$ and $\eta_3$. The total edge length from the root to any leaf is always equal to 1. For example, $t_1+t_2+\eta_3=1$. Thus $S(\tree)\leq p$.
Conditional probabilities \eqref{eq:pIBP_cond_probs}  imply that once the variable at a node becomes $1$, all its child nodes (including leaf nodes) become 1. Solid circles represent that the variables at the corresponding positions are 1. Hollow circles denote 0's. Bold segments are drawn for the edges connecting 1's. 
The leaves are divided into 6 subgroups according to their common parent nodes $\zeta_l$'s.
}
\label{fig:eg_tree2}
\end{figure}

In Step 2, for $Z \in \zsobm_p$ with $K$ columns,
\begin{align*}
\Pi_{\tilde{K}}(Z  \mid  \alpha, \tree) \triangleq \sum_{\tilde{Z} \in \tsobm_{p \times \tilde{K}}: G(\tilde{Z}) = Z} P_{\tilde{K}} (\tilde{Z} \mid  \alpha, \tree) = \binom{\tilde{K}}{K}P_{\tilde{K}}(\tilde{Z}_* \mid \alpha, \tree),
\end{align*}
where $\tilde{Z}_*$ is such that the first $K$ columns of $\tilde{Z}_*$ are exactly the same with $Z$ and the rest $\tilde{K}-K$ columns are all $\mathbf{0}$. For notational simplicity we still denote $\tilde{Z}_*$ by $\tilde{Z}$. 

In Step 3, to obtain the pmf $\Pi(Z \mid \alpha, \tree)$, we take the limit $\lim_{\tilde{K} \rightarrow \infty} \Pi_{\tilde{K}}(Z  \mid  \alpha, \tree)$.
We consider two cases (1) $Z \in \sobm_p$ and (2) $Z = \bm 0$.

We first consider $Z \in \sobm_p$.
Denote by $Z = (\bm{z}_1,\bm{z}_2,\cdots,\bm{z}_K)$, we have
\begin{align*}
&\Pi_{\tilde{K}}(Z \mid \alpha, \tree)  \\
= {}&\frac{\tilde{K}!}{K!(\tilde{K}-K)!}P_{\tilde{K}}(\tilde{Z} \mid \alpha, \tree)  \\
= {}&\underbrace{\frac{\tilde{K}!}{K!(\tilde{K}-K)!}\prod\limits_{k=1}^{K}\int P(\bm{z}_k \mid \pi_k, \tree) p(\pi_k \mid \alpha)\text{d}\pi_k}_{P_{1, \tilde{K}}(\tilde{Z} \mid \alpha, \tree)} \cdot
\underbrace{\prod\limits_{k>K}^{\tilde{K}}\int P(\mathbf{0} \mid \pi_k, \tree) p(\pi_k \mid \alpha)\text{d}\pi_k}_{P_{2, \tilde{K}}(\tilde{Z} \mid \alpha, \tree)}.
\end{align*}

We consider $P_{1, \tilde{K}}(\tilde{Z} \mid \alpha, \tree)$ and $P_{2, \tilde{K}}(\tilde{Z} \mid \alpha, \tree)$ separately. For $P_{1, \tilde{K}}(\tilde{Z} \mid \alpha, \tree)$, we define 
\begin{align}
\lambda(\bm{z}_k, \tree) \triangleq 
\lim_{\tilde{K} \rightarrow \infty} \int_0^1 P(\bm{z}_k \mid \pi_k, \tree)\pi_k^{\frac{\alpha}{\tilde{K}}-1}\text{d}\pi_k  = 
\int_0^1P(\bm{z}_k \mid \pi_k, \tree)\pi_k^{-1}\text{d}\pi_k,
\label{eq:lambda_k}
\end{align}
where $\lambda(\bm{z}_k, \tree) < \infty$ for $\bm{z}_k \neq \bm 0$.
Thus,
\begin{align}
\lim_{\tilde{K} \rightarrow \infty} P_{1, \tilde{K}}(\tilde{Z} \mid \alpha, \tree) 
&= \lim_{\tilde{K} \rightarrow \infty} \frac{\tilde{K}!}{K!(\tilde{K}-K)!}\cdot (\frac{\alpha}{\tilde{K}})^{K}\cdot \prod\limits_{k=1}\limits^{K}\int_0^1P(\bm{z}_k \mid \pi_k, \tree)\pi_k^{\frac{\alpha}{\tilde{K}}-1}\text{d}\pi_k  \notag \\
&= \lim_{\tilde{K} \rightarrow \infty} \frac{\tilde{K}!}{(\tilde{K}-K)!\tilde{K}^{K}}\cdot\frac{\alpha^K}{K!}\cdot\prod\limits_{k=1}\limits^{K}\int_0^1 P(\bm{z}_k \mid \pi_k, \tree)\pi_k^{\frac{\alpha}{\tilde{K}}-1}\text{d}\pi_k  \notag\\
& = 1 \cdot \frac{\alpha^K}{K!} \cdot \prod\limits_{k=1}\limits^{K} \lambda(\bm{z}_k, \tree). \label{eq:P_1_K_tilde}
\end{align}

For $P_{2, \tilde{K}}(\tilde{Z} \mid \alpha, \tree)$, we need to calculate $P(\bm 0 \mid \pi_k, \tree)$. In the case that the variables at the leaves $\bm{z}_k = \bm 0$, the random variables at all nodes of the tree must take value $0$ (see Figure \ref{fig:eg_tree2}). Based on equations \eqref{eq:pIBP_cond_probs}, we calculate
\begin{align*}
P(\mathbf{0} \mid \pi_k, \tree)=\prod_{t \in \text{edges of $\tree$}} \exp(-\gamma_k t)=\exp(-\gamma_k \sum_{t \in \text{edges of $\tree$}} t)=(1-\pi_k)^{S(\tree)}
\end{align*}
where $\gamma_k = -\log(1 - \pi_k)$ and $S(\tree)$ is defined as the total edge lengths of $\tree$. Accordingly,
\begin{align*}
\int P(\mathbf{0} \mid \pi_k, \tree) p(\pi_k \mid \alpha)\text{d}\pi_k =
\int P(\mathbf{0} \mid \pi, \tree) p(\pi \mid \alpha)\text{d}\pi=\frac{\alpha}{\tilde{K}}\int_0^1 (1-\pi)^{S(\tree)}\pi^{\frac{\alpha}{\tilde{K}}-1}\text{d}\pi.
\end{align*}
Thus
\begin{align*}
P_{2, \tilde{K}}(\tilde{Z} \mid \alpha, \tree) &= 
\Big(\int P(\mathbf{0} \mid \pi, \tree) p(\pi \mid \alpha)\text{d}\pi\Big)^{\tilde{K}-K} \\
&=\Big(\frac{\alpha}{\tilde{K}}\int_0^1 (1-\pi)^{S(\tree)}\pi^{\frac{\alpha}{\tilde{K}}-1}\text{d}\pi\Big)^{\tilde{K}-K}\\
&=\Big[\frac{\alpha}{\tilde{K}}\text{Beta}\Big(\frac{\alpha}{\tilde{K}},S(\tree)+1\Big)\Big]^{\tilde{K}-K}\\
&=\Big[\frac{\Gamma(\frac{\alpha}{\tilde{K}}+1)\Gamma(S(\tree)+1)}{\Gamma(\frac{\alpha}{\tilde{K}}+S(\tree)+1)}\Big]^{\tilde{K}-K}\\
&=\exp\Big[(\tilde{K}-K)\log\Big(\frac{\Gamma(\frac{\alpha}{\tilde{K}}+1)\Gamma(S(\tree)+1)}{\Gamma(\frac{\alpha}{\tilde{K}}+S(\tree)+1)}\Big)\Big],
\end{align*}
where
\begin{align*}
\log(\frac{\Gamma(\frac{\alpha}{\tilde{K}}+1)\Gamma(S(\tree)+1)}{\Gamma(\frac{\alpha}{\tilde{K}}+S(\tree)+1)})
=&\log(1+\frac{\Gamma(\frac{\alpha}{\tilde{K}}+1)\Gamma(S(\tree)+1)-\Gamma(\frac{\alpha}{\tilde{K}}+S(\tree)+1)}{\Gamma(\frac{\alpha}{\tilde{K}}+S(\tree)+1)})\\
\sim&\frac{\Gamma(\frac{\alpha}{\tilde{K}}+1)\Gamma(S(\tree)+1)-\Gamma(\frac{\alpha}{\tilde{K}}+S(\tree)+1)}{\Gamma(\frac{\alpha}{\tilde{K}}+S(\tree)+1)}\\
\sim&\frac{\Gamma(\frac{\alpha}{\tilde{K}}+1)\Gamma(S(\tree)+1)-\Gamma(\frac{\alpha}{\tilde{K}}+S(\tree)+1)}{\Gamma(S(\tree)+1)}\\
=&\Gamma(\frac{\alpha}{\tilde{K}}+1)-\frac{\Gamma(\frac{\alpha}{\tilde{K}}+S(\tree)+1)}{\Gamma(S(\tree)+1)}.
\end{align*}
Here $A(\tilde{K}) \sim B(\tilde{K})$ means that $A(\tilde{K}) / B(\tilde{K})\rightarrow 1$ as $\tilde{K} \rightarrow \infty$.
Recall that $\psi(x)=\Gamma^{'}(x) / \Gamma(x)$ denote the digamma function. We have
\begin{align*}
&\Gamma(\frac{\alpha}{\tilde{K}}+1)-\frac{\Gamma(\frac{\alpha}{\tilde{K}}+S(\tree)+1)}{\Gamma(S(\tree)+1)}=-\left(\psi(S(\tree)+1)-\psi(1)\right)\frac{\alpha}{\tilde{K}}+o(\frac{1}{\tilde{K}}).
\end{align*}
Finally,
\begin{align}
\lim_{\tilde{K} \rightarrow \infty} P_{2, \tilde{K}}(\tilde{Z} \mid \alpha, \tree) &= 
\lim_{\tilde{K} \rightarrow \infty} \Big(\int P(\mathbf{0} \mid \pi, \tree) p(\pi \mid \alpha)\text{d}\pi\Big)^{\tilde{K}-K} 
\notag \\
&= \lim_{\tilde{K} \rightarrow \infty} \exp\Big[(\tilde{K}-K)\log\Big(\frac{\Gamma(\frac{\alpha}{\tilde{K}}+1)\Gamma(S(\tree)+1)}{\Gamma(\frac{\alpha}{\tilde{K}}+S(\tree)+1)}\Big)\Big]
\notag \\
&= \exp\{-\left(\psi(S(\tree)+1)-\psi(1)\right)\alpha\}. \label{eq:P_2_K_tilde}
\end{align}
as $\tilde{K}\rightarrow \infty$.

Combining equations \eqref{eq:P_1_K_tilde} and \eqref{eq:P_2_K_tilde}, we get
\begin{align*}
\Pi(Z \mid \alpha, \tree) &= \lim_{\tilde{K} \rightarrow \infty} \Pi_{\tilde{K}}(Z  \mid  \alpha, \tree) \\
&= \exp\{-\left(\psi(S(\tree)+1)-\psi(1)\right)\alpha\}  \frac{\alpha^K}{K!}  \prod\limits_{k=1}\limits^{K} \lambda(\bm{z}_k, \tree),
\end{align*}
for $Z \in \sobm_p$.

We then consider $Z = \bm 0$, in which case $\tilde{Z}$ is a matrix with $\tilde{K}$ all zero columns. We have
\begin{align*}
\Pi_{\tilde{K}}(Z \mid \alpha, \tree)=&\prod\limits_{k=1}^{\tilde{K}}\int P(\mathbf{0} \mid \pi_k, \tree) p(\pi_k \mid \alpha) \text{d}\pi_k\\
=&\Big(\int P(\mathbf{0} \mid \pi, \tree) p(\pi \mid \alpha)\text{d}\pi\Big)^{\tilde{K}}.
\end{align*}
According to equation \eqref{eq:P_2_K_tilde}, we get
\begin{align*}
\Pi(Z \mid \alpha, \tree) &= \lim_{\tilde{K} \rightarrow \infty} \Pi_{\tilde{K}}(Z  \mid  \alpha, \tree) \\
&= \exp\{-\left(\psi(S(\tree)+1)-\psi(1)\right)\alpha\} ,
\end{align*}
for $Z = \bm 0$.

To summarize,

\begin{align*}
\Pi(Z \mid \alpha, \tree) =\begin{cases}
\exp\{-\left(\psi(S(\tree)+1)-\psi(1)\right)\alpha\}, &Z=\bm{0};\\ 
\exp\{-\left(\psi(S(\tree)+1)-\psi(1)\right)\alpha\}\frac{\alpha^K}{K!}\prod\limits_{k=1}^K\lambda(\bm{z}_k, \tree), & Z\in \sobm_p.
\end{cases}
\end{align*}

\subsection{Proof of Lemma \ref{lem1}}
\label{app:lower_bound}

For the proof of Lemma \ref{lem1}, we first introduce some definition and notation. 
We say two leaf nodes belong to the same \emph{subgroup}, if they share the same parent node.
Let $L$ denote the total number of subgroups. Denote by node$_l$ the common parent node for subgroup $l$, $M_l$ the number of leaf nodes in subgroup $l$, and $\eta_l$ the length between node$_l$ and any leaf node in subgroup $l$, $l = 1, \ldots, L$. 
We have $p = \sum_{l = 1}^L M_l$.  
Furthermore, denote by $m_{kl}$ the number of leaf nodes in column $k$ that belong to subgroup $l$ and equal to 1, i.e. $m_{kl} = \sum_{j \in  \text{subgroup $l$}} z_{jk}$. We have $\sum_{l = 1}^L m_{kl} = m_k$. See Figure \ref{fig:eg_tree2} for an example.


Next, we give a preparatory lemma.
\begin{lemma}
\label{lem:lower_bound_lambda_k}
Consider $\lambda_k = \lambda(\bm{z}_k, \tree) = \int_{0}^1 P(\bm z_k \mid \pi_k, \tree) \pi_k^{-1} d\pi_k$ (see Equation \eqref{eq:lambda_k}). Each $\lambda_k$ has the following lower bound
\begin{align*}
\lambda_k \geq \int_0^1 (1-\pi_k)^{S(\tree)-\sum\limits_{l}\eta_l m_{kl}} \prod_l \left( [1-(1-\pi_k)^{\eta_l}]^{m_{kl}}\right) \pi_k^{-1}\text{d}\pi_k,
\end{align*}
where $S(\tree)$ denotes the total length of tree structure $\tree$ (see example in Figure \ref{fig:eg_tree}).
\end{lemma}

\begin{proof}[Proof of Lemma \ref{lem:lower_bound_lambda_k}]

Let $\Omega$ denote the event that the variables at all interior nodes of the tree take value $0$. For example, in Figure \ref{fig:eg_tree2}, all interior nodes (in blue and red) should take value $0$ (hollow circle). We have 
\begin{align*}
P(\bm{z}_k \mid \pi_k, \tree) &\geq  P(\Omega \mid \pi_k, \tree)\times P(\bm{z}_k \mid \Omega, \pi_k, \tree)\\
& = (1-\pi_k)^{S(\tree)-\sum\limits_{l}\eta_l M_l} \times
\prod_{l}\Big([1-(1-\pi_k)^{\eta_l}]^{m_{kl}} 
(1 - \pi_k)^{\eta_l(M_l - m_{kl})}\Big)\\
& = (1-\pi_k)^{ S(\tree) - \sum\limits_{l}\eta_l m_{kl}}\prod\limits_l \Big( [1-(1-\pi_k)^{\eta_l}]^{m_{kl}}\Big),
\end{align*}
which implies that
\begin{align*}
\lambda_k \geq \int_0^1 (1-\pi_k)^{S(\tree)-\sum\limits_{l}\eta_l m_{kl}}\prod\limits_l \left( [1-(1-\pi_k)^{\eta_l}]^{m_{kl}}\right) \pi_k^{-1}\text{d}\pi_k.
\end{align*}
\end{proof}

\begin{proof}[Proof of Lemma \ref{lem1}]
For notational simplicity, we omit the index $n$ in $K_{n}^{*}$, $p_n$, $s_n$, $m_{kn}$, $Z_n$, $Z_{n}^{*}$, etc., in the following proof. We use $C$, $C_1, C_2, \ldots, C_{7}$ to represent positive constants.
It suffices to prove the result for pIBP, as IBP is a special case of pIBP, in which the length between each leaf node and its parent node (the root) is always $1$.

We first consider $Z^* \in \sobm_{p \times K^{*}}$, \ie, none of the columns of $Z^*$ consist of all 0's.
Based on the assumption in case (2), $\eta_l \geq \eta_0$, $\forall\ 1\leq l\leq L$ and $\forall$ $n$. Lemma \ref{lem:lower_bound_lambda_k} implies that
\begin{align*}
\lambda_k \geq &\int_0^1 (1-\pi_k)^{S(\tree)-\sum\limits_{l}\eta_l m_{kl}} [1-(1-\pi_k)^{\eta_0}]^{\sum\limits_{l} m_{kl}}\pi_k^{-1}\text{d}\pi_k\\
\geq &\int_0^1 (1-\pi_k)^{S(\tree)} [1-(1-\pi_k)^{\eta_0}]^{m_k}\pi_k^{-1} \text{d}\pi_k \quad \left( \text{Let } u = (1-\pi_k)^{1-\eta_0} \right)\\
= {}&\int_0^1 u^{\frac{S(\tree) }{1-\eta_0}}(1-u)^{m_k}(1-u^{\frac{1}{1-\eta_0}})^{-1}\frac{1}{1-\eta_0}u^{\frac{\eta_0}{1-\eta_0}}\text{d}u\\
\geq {}&\frac{1}{1-\eta_0}\int_0^1 u^{C_1 S(\tree)}(1-u)^{m_k}\text{d}u \\
\geq {}&\frac{1}{1-\eta_0}\int_0^1 u^{C_1 p}(1-u)^{m_k}\text{d}u,
\end{align*}
for some positive constant $C_1$, considering that $S(\tree)\leq p$ (see Figure \ref{fig:eg_tree2} for an explanation).

Given $m_k\leq s$, $\int_0^1 u^{C_1p}(1-u)^{m_k}\text{d}u\geq \int_0^1 u^{C_1p}(1-u)^{s}\text{d}u$. Thus we have
\begin{align*}
\lambda_k\geq C_2 \int_0^1 u^{C_1 p}(1-u)^{s}\text{d}u = C_2 \ \text{Beta}(C_1 p+1,s+1)=C_2 \frac{\Gamma(C_1 p+1)\Gamma(s+1)}{\Gamma(C_1 p+s+2)}.
\end{align*}

By Stirling's formula, we get
\begin{align*}
\lambda_k \geq&C_2 \frac{\Gamma(C_1 p+1)\Gamma(s+1)}{\Gamma(C_1 p+s+2)}\\
\sim {}& C_2 \frac{\sqrt{2\pi C_1 p}\left( \frac{C_1 p}{e}\right) ^{C_1 p}\sqrt{2\pi s}\left( \frac{s}{e}\right) ^{s}}{\sqrt{2\pi (C_1 p+s+1)}\left( \frac{C_1 p+s+1}{e}\right) ^{C_1 p+s+1}}\\
\sim {}& C_3\frac{(C_1 p)^{C_1 p}s^{s+\frac{1}{2}}}{(C_1 p+s+1)^{C_1 p+s+1}} \\
\geq {}& C_3 \left( \frac{C_1 p}{C_1 p+s+1}\right) ^{C_1 p}\left( \frac{1}{C_1 p+s+1}\right)^{s+1},
\end{align*}
considering $s/p\rightarrow 0$ and $s \geq 1$.
Since $s/p\rightarrow 0$, $\exists\ C_4 > 0$ and $n_0\in \mathbb{N}^+$ such that $\log\left( 1-\frac{s+1}{C_1 p+s+1}\right)\geq -C_4 \frac{s+1}{C_1 p+s+1}$ as $n\geq n_0$. So
\begin{align*}
\left( \frac{C_1 p}{C_1 p+s+1}\right) ^{C_1p} = \exp\left( C_1p\log\left( 1-\frac{s+1}{C_1 p+s+1}\right) \right)\geq\exp(-C_5 s).
\end{align*}
Thus,
\begin{align*}
\lambda_k\geq C_3 \exp(-C_5 s)\left( \frac{1}{C_1 p+s+1}\right)^{s+1}\geq \exp(-C_6 s\log(p+1)),
\end{align*}
for $n\geq n_0$. And $\lambda_k\geq \exp(-C_6 s\log(p+1))$ also holds for $n< n_0$ by choosing a  sufficiently large $C_6$.

For pIBP mixture,
\begin{align*}
\Pi(Z=Z^*)=\Big( \psi\Big( S(\tree)+1\Big)-\psi(1)+1\Big) ^{-(K^*+1)}\prod_{k=1}^{K^*}\lambda_k,
\end{align*}
and for the digamma function $\psi(\cdot)$, we have
\begin{align*}
\psi(S(\tree)+1)-\psi(1)+1=\log(S(\tree))+O(1)\leq \log p + O(1).
\end{align*}
Therefore, $$\log (\psi(S(\tree)+1)-\psi(1)+1)\leq C_7\max\{\log \log (p+1),1\}.$$
Thus,
\begin{align*}
(\psi(S(\tree)+1)-\psi(1)+1) ^{-(K^*+1)}= {}& \exp\{-(K^*+1)\log(\psi(S(\tree)+1)-\psi(1)+1)\}\\
\geq {}& \exp(-2C_7 K^*\max\{\log \log (p+1),1\}),
\end{align*}
and we finally have
\begin{align}
\Pi(Z=Z^*)\geq &\exp(-2C_7 K^*\max\{\log \log (p+1),1\})\left( \exp(-C_6 s\log (p+1))\right) ^{K^*} \nonumber\\
\geq &\exp(-CsK^*\log (p+1)), \label{eq:lower_bound}
\end{align}
for some positive constant $C$.

Next, we consider $Z_{n}^{*} = \bm 0$.
We immediately have 
\begin{align*}
\Pi(Z_n=Z_{n}^{*}) =&\Big(\psi(S(\tree)+1)-\psi(1)+1\Big)^{-1}\\\geq  &\exp(-C_7 \max\{\log \log (p+1),1\})\\
\geq &\exp(-C\log(p+1)).
\end{align*}

\end{proof}

\subsection{Proof of Theorem \ref{thm1}}
For the proof of Theorem \ref{thm1}, we first introduce two preparatory lemmas from the literature.

\begin{lemma}[Theorem 1 in \citealt{chen2016posterior}]\label{Lem1}
Let $\sobm$ be a collection of binary matrices that contains $Z^*$. Consider a family of probability measures indexed by $Z \in \sobm$, i.e., $\{\mu_{Z}: Z \in \sobm\}$. For any subset $U$ of $\sobm$ and any testing function $\phi$ based on $X$,
\begin{align*}
E_{Z^*}[\Pi(U \mid X)]\leq E_{Z^*} (\phi) + \frac{1}{\Pi(ZZ^{T}=Z^*Z^{*T})}\sup\limits_{Z\in U}E_Z
(1-\phi),
\end{align*}
where $E_Z(\cdot)$ means taking expectation in the case where $X\sim\mu_Z$.
\end{lemma}

\begin{lemma}[Remark 5.40.2 in \citealt{vershynin2012introduction}]\label{Lem2}
Suppose that $p$ dimensional random variables $\bx_i^{T}\stackrel{i.i.d.}{\sim} N(\bm 0, \Sigma)$ for $i=1,2,\cdots,n$ and let $\hat{\Sigma}=\sum\limits_{i=1}\limits^n \bx_i^T \bx_i/n$ denote the sample covariance matrix. Then for every $t \geq 0$,
\begin{align*}
P_{\Sigma}(\|\hat{\Sigma}-\Sigma\|\geq \max\{\delta,\delta^2\}\|\Sigma\|)\leq \exp(-Ct^2),
\end{align*}
where $\delta = C' \sqrt{\frac{p}{n}}+\frac{t}{\sqrt{n}}$ and $C$ and $C'$ are positive constants.
\end{lemma}

\begin{proof}[Proof of Theorem \ref{thm1}]
Let $\Sigma_n=Z_nZ_n^T+\bm{I}_{p_n}$ and $\Sigma_{n}^* =Z_{n}^{*}Z_{n}^{*T}+\bm{I}_{p_n}$ denote the model covariance matrix and the true covariance matrix, respectively.
For any $\tilde{Z}_{n}^* \in \tsobm_{p_n \times \tilde{K}_{n}^*}$ and $Z_{n}^{*} = G(\tilde{Z}_{n}^*) \in \sobm_{p_n}^0$, we have
\begin{enumerate}
\item $\{Z_n:\|Z_nZ_n^T-\tilde{Z}_{n}^* \tilde{Z}_{n}^{*T} \|\leq C\epsilon_n\} = \{Z_n:\|Z_nZ_n^T-Z_{n}^{*}Z_{n}^{*T}\|\leq C\epsilon_n\}$;
\item $\text{E}_{\tilde{Z}_{n}^*}(\cdot) = \text{E}_{Z_{n}^{*}}(\cdot)$, considering the binary factor model where 
\begin{align*}
\bx_i^{T} \mid Z_{n}^{*}\stackrel{i.i.d}{\sim} N(\bm 0, Z_{n}^{*}Z_{n}^{*T}+\bm{I}_{p_n}).
\end{align*}
\end{enumerate}
Based on the above two facts, it suffices to prove for $Z_{n}^* \in \sobm_{p_n}^0$.
In the following proof, $C$ represents the positive constant in Theorem \ref{thm1}, and $C'$, $C_1, C_2, \ldots, C_{11}$ represent other positive constants.

We first consider the case that $Z_n^* \in \sobm_{p_n \times K_{n}^{*}}$. That is, $m_{kn}\geq 1$ for all $k = 1, \ldots, K_{n}^{*}$ and $n = 1, 2, \ldots$. 
Recall Lemma \ref{lem1},
\begin{align}
\Pi(Z_n = Z_{n}^{*} )\geq \exp(-C_1 s_n K_{n}^{*} \log (p_n+1)).
\label{cond0_pf_thm}
\end{align}
Also, recall the assumption 
$$
\epsilon_n=\frac{\max\{\sqrt{p_n},\sqrt{s_nK_{n}^{*}\log (p_n+1)}\}}{\sqrt{n}}\|Z_{n}^{*}Z_{n}^{*T}\|  \rightarrow 0,
$$ 
as $n\rightarrow \infty$. Since $\|\Sigma_{n}^* \|=\|Z_{n}^{*} Z_{n}^{*T}+\bm{I}_{p_n}\|=\|Z_{n}^{*} Z_{n}^{*T}\|+1 \leq C_2 \|Z_{n}^{*} Z_{n}^{*T}\|$, the following three conditions hold:
\begin{align}
&\epsilon_n' \triangleq \frac{\max\{\sqrt{p_n},\sqrt{s_n K_{n}^{*} \log(p_n + 1)}\}}{\sqrt{n}}\|\Sigma_{n}^* \| \rightarrow 0, \quad \text{as $n\rightarrow \infty$}; \label{cond1_pf_thm}\\
&\frac{p_n}{n} \rightarrow 0, \quad \text{as $n\rightarrow \infty$}; \label{cond2_pf_thm} \\
&\frac{s_n K_{n}^{*} \log(p_n + 1)}{n} \rightarrow 0, \quad \text{as $n\rightarrow \infty$}. \label{cond3_pf_thm}
\end{align} 
Condition \eqref{cond1_pf_thm} is due to $\epsilon_n' \leq C_2 \epsilon_n \rightarrow 0$ as $n\rightarrow \infty$, condition \eqref{cond2_pf_thm} is due to $\sqrt{p_n / n} \leq \sqrt{p_n / n} \left( \|Z^*Z^{*T}\|+1 \right) \leq \epsilon_n' \rightarrow 0$ as $n\rightarrow \infty$, and condition \eqref{cond3_pf_thm} is due to $\sqrt{s_n K_{n}^{*} \log(p_n + 1) / n} \leq \sqrt{s_n K_{n}^{*} \log(p_n + 1) / n} \left( \|Z^*Z^{*T}\|+1 \right) \leq \epsilon_n' \rightarrow 0$ as $n\rightarrow \infty$.

Since $C_2 \epsilon_n \geq \epsilon_n'$, we can re-write $C \epsilon_n \geq C' \epsilon_n'$ for some positive constant $C' \leq C / C_2$. To show
\begin{align*}
E_{Z_{n}^{*}}[\Pi(\|Z_nZ_n^T-Z_{n}^{*}Z_{n}^{*T}\| \leq C \epsilon_n \mid X)]\rightarrow 1, \quad \text{as $n \rightarrow \infty$},
\end{align*}
it suffices to show
\begin{align*}
E_{Z_{n}^{*}}[\Pi(\|Z_nZ_n^T-Z_{n}^{*}Z_{n}^{*T}\| \geq C \epsilon_n \mid X)]\rightarrow 0, \quad \text{as $n \rightarrow \infty$},
\end{align*}
or
\begin{align*}
E_{Z_{n}^{*}}[\Pi(\|Z_nZ_n^T-Z_{n}^{*}Z_{n}^{*T}\| \geq C' \epsilon_n' \mid X)]\rightarrow 0, \quad \text{as $n \rightarrow \infty$}.
\end{align*}

For notational simplicity, we omit the index $n$ in $K_{n}^{*}$, $p_n$, $s_n$, $m_{kn}$, $Z_n$, $Z_{n}^{*}$, $\Sigma_n$, $\Sigma_{n}^*$, etc., in the following proof.
Let
\begin{align}
U\triangleq\{\|ZZ^T-Z^*Z^{*T}\|\geq C' \epsilon'\}=\{\|\Sigma-\Sigma^*\|\geq C' \epsilon'\}.
\label{u}
\end{align}
Thus,
\begin{align*}
E_{Z^*}[\Pi(U \mid X)] = E_{Z^*}[\Pi(\|ZZ^T-Z^*Z^{*T}\|\geq C' \epsilon' \mid X)].
\end{align*}

Lemma \ref{Lem1} implies that to show $E_{Z^*}[\Pi(U \mid X)] \rightarrow 0$ as $n \rightarrow \infty$, it suffices to show  
$\frac{1}{\Pi(ZZ^{T}=Z^*Z^{*T})}\sup\limits_{Z\in U} E_Z(1-\phi) \rightarrow 0$
and
$E_{Z^*}\phi \rightarrow 0$ as $n \rightarrow \infty$,
for some testing function $\phi$.

Define the testing function 
\begin{align*}
\phi=\mathds{1}_{\{\|\hat{\Sigma}-\Sigma^*\|> \frac{C' \epsilon'}{2}\}}.
\end{align*}
We first show $\lim\limits_{n \rightarrow \infty} \frac{1}{\Pi(ZZ^{T}=Z^*Z^{*T})}\sup\limits_{Z\in U} E_Z(1-\phi) = 0$, where $E_Z(1-\phi)=P_\Sigma(\|\hat{\Sigma}-\Sigma^*\|\leq \frac{C' \epsilon'}{2})$. 
We consider two scenarios
$\|\Sigma\|\geq 2\|\Sigma^*\|$ and $\|\Sigma\|< 2\|\Sigma^*\|$ separately. 

For  $\|\Sigma\|\geq 2\|\Sigma^*\|$,
\begin{align*}
P_\Sigma(\|\hat{\Sigma}-\Sigma^*\|\leq \frac{C' \epsilon'}{2})&\leq P_\Sigma(\|\hat{\Sigma}-\Sigma\|\geq \|\Sigma\|-\|\Sigma^*\|-\frac{C' \epsilon'}{2})\\
&=P_\Sigma\left( 
\frac{\|\hat{\Sigma}-\Sigma\|}{\|\Sigma\|}\geq 1-\frac{\|\Sigma^*\|}{\|\Sigma\|}-\frac{C' \epsilon'}{2\|\Sigma\|}\right)\\
&\leq P_{\Sigma}\left( \frac{\|\hat{\Sigma}-\Sigma\|}{\|\Sigma\|}\geq \frac{1}{2}-\frac{C' \epsilon'}{2\|\Sigma\|}\right)\\
&\leq  P_{\Sigma}\left( \frac{\|\hat{\Sigma}-\Sigma\|}{\|\Sigma\|}\geq \frac{1}{2}-\frac{C' \epsilon'}{4\|\Sigma^*\|}\right)\\
&\leq P_{\Sigma}\left( \frac{\|\hat{\Sigma}-\Sigma\|}{\|\Sigma\|}\geq \frac{1}{2}-\frac{C' \epsilon'}{4}\right).
\end{align*}
Since $\epsilon' \rightarrow 0$ as $n\rightarrow \infty$, for sufficiently large $n$, we have
\begin{align*}
P_{\Sigma}\left( \frac{\|\hat{\Sigma}-\Sigma\|}{\|\Sigma\|}\geq \frac{1}{2}-\frac{C' \epsilon'}{4}\right)\leq  P_{\Sigma}\left( \frac{\|\hat{\Sigma}-\Sigma\|}{\|\Sigma\|}\geq \frac{1}{4}\right).
\end{align*}
Applying Lemma \ref{Lem2} with $t= \sqrt{n} / 5$. Since $p/n\rightarrow 0$ as $n\rightarrow \infty$, $\delta = C_3 \sqrt{\frac{p}{n}}+\frac{t}{\sqrt{n}}\leq \frac{1}{4}$ for sufficiently large $n$. Thus
\begin{align*}
P_{\Sigma}\left( \frac{\|\hat{\Sigma}-\Sigma\|}{\|\Sigma\|}\geq \frac{1}{4}\right) \leq P_{\Sigma}\left( \frac{\|\hat{\Sigma}-\Sigma\|}{\|\Sigma\|}\geq \max\{\delta,\delta^2\}\right)\leq \exp(-C_4 n).
\end{align*}
Applying Lemma \ref{lem1} \eqref{cond0_pf_thm} and since $[s K^* \log (p+1)] / n \rightarrow 0$ as $n\rightarrow\infty$,
\begin{align*}
\frac{1}{\Pi(ZZ^{T}=Z^*Z^{*T})}\sup\limits_{\|\Sigma\|\geq 2\|\Sigma^*\|}P_\Sigma(\|\hat{\Sigma}-\Sigma^*\|\leq \frac{C' \epsilon'}{2})\leq \frac{\exp(-C_4 n)}{\exp(-C_1 sK^*\log (p+1))}\rightarrow 0,
\end{align*}
as $n\rightarrow \infty$.

Next, we consider $\|\Sigma\|\leq 2\|\Sigma^*\|$. Applying Lemma \ref{Lem2} with $t = C_5 \sqrt{sK^*\log (p+1)}$,
\begin{align*}
P_{\Sigma}(\|\hat{\Sigma}-\Sigma\|\geq 2\max\{\delta,\delta^2\}\|\Sigma^*\|)
&\leq P_{\Sigma}(\|\hat{\Sigma}-\Sigma\|\geq \max\{\delta,\delta^2\}\|\Sigma\|) \\
&\leq \exp(-C_6 C_5^2 sK^*\log (p+1)),
\end{align*}
where $\delta = C_7 \sqrt{\frac{p}{n}}+ C_5 \sqrt{\frac{sK^*\log (p+1)}{n}}$. 
Again, according to Lemma \ref{lem1} \eqref{cond0_pf_thm}, 
\begin{multline*}
\frac{1}{\Pi(ZZ^{T}=Z^*Z^{*T})}\sup\limits_{\|\Sigma\|\leq 2\|\Sigma^*\|}P_{\Sigma}(\|\hat{\Sigma}-\Sigma\|\geq 2\max\{\delta,\delta^2\}\|\Sigma^*\|) \leq \\
\frac{\exp(-C_6 C_5^2 sK^*\log (p+1))}{\exp(-C_1 sK^*\log (p+1))} \rightarrow 0,
\end{multline*}
for a sufficiently large $C_5 > 0$.
Due to conditions \eqref{cond2_pf_thm} and \eqref{cond3_pf_thm},
$\delta\rightarrow 0$ as $n \rightarrow \infty$. Thus, for sufficiently large $n$,
\begin{align*}
\max\{\delta,\delta^2\}=\delta \leq \frac{C_8}{2} \cdot \frac{\max\{\sqrt{p},\sqrt{sK^*\log (p+1)}\}}{\sqrt{n}},
\end{align*}
for some positive constant $C_8$, which implies
\begin{multline*}
\frac{1}{\Pi(ZZ^{T}=Z^*Z^{*T})}\sup\limits_{\|\Sigma\|\leq 2\|\Sigma^*\|}P_{\Sigma}\Big( \|\hat{\Sigma}-\Sigma\|\geq C_8 \frac{\max\{\sqrt{p},\sqrt{sK^*\log (p+1)}\}}{\sqrt{n}}\|\Sigma^*\|\Big) = \\
\frac{1}{\Pi(ZZ^{T}=Z^*Z^{*T})}\sup\limits_{\|\Sigma\|\leq 2\|\Sigma^*\|}P_{\Sigma}\Big( \|\hat{\Sigma}-\Sigma\|\geq C_8 \epsilon' \Big)
\rightarrow 0,
\end{multline*}
as $n\rightarrow \infty$. Recall that $\phi$ is the indicator function $\mathds{1}_{\{\|\hat{\Sigma}-\Sigma^*\|> \frac{C' \epsilon'}{2}\}}$ and $U=\{\|\Sigma-\Sigma^*\|\geq C' \epsilon'\}$ as in (\ref{u}). By choosing $C' \geq 2 C_8$, we have
\begin{align*}
{}&\frac{1}{\Pi(ZZ^{T}=Z^*Z^{*T})}\sup\limits_{Z\in U,\ \|\Sigma\|\leq 2\|\Sigma^*\|}E_Z (1-\phi) \\
= {}&\frac{1}{\Pi(ZZ^{T}=Z^*Z^{*T})}\sup\limits_{Z\in U,\ \|\Sigma\|\leq 2\|\Sigma^*\|}P_{\Sigma}\left( \|\hat{\Sigma}-\Sigma^*\|\leq \frac{C' \epsilon'}{2}\right) \\
\leq {}&\frac{1}{\Pi(ZZ^{T}=Z^*Z^{*T})}\sup\limits_{\|\Sigma\|\leq 2\|\Sigma^*\|}P_{\Sigma}\left( \|\hat{\Sigma}-\Sigma\|\geq \frac{C' \epsilon'}{2}\right) \\
\leq {}&\frac{1}{\Pi(ZZ^{T}=Z^*Z^{*T})}\sup\limits_{\|\Sigma\|\leq 2\|\Sigma^*\|}P_{\Sigma}\Big( \|\hat{\Sigma}-\Sigma\|\geq C_8 \epsilon' \Big)
\rightarrow 0,
\end{align*}
as $n\rightarrow \infty$.

Combining the results for $\|\Sigma\|\geq 2\|\Sigma^*\|$ and $\|\Sigma\|\leq 2\|\Sigma^*\|$, we have proved
\begin{align}
\frac{1}{\Pi(ZZ^{T}=Z^*Z^{*T})}\sup\limits_{Z\in U}E_Z
(1-\phi)\rightarrow 0, \label{Lim1}
\end{align}
as $n\rightarrow \infty$, for $C' \geq 2C_8$.

Next, it remains to show $\lim\limits_{n \rightarrow \infty} E_{Z^*} \phi = 0$. 

Applying Lemma \ref{Lem2} with $t = \sqrt{p}$ and $\delta = C_9 \sqrt{\frac{p}{n}} + \frac{t}{\sqrt{n}}= C_{10} \sqrt{\frac{p}{n}}$. 
Due to condition \eqref{cond2_pf_thm}, $\delta\rightarrow 0$ as $n\rightarrow \infty$.
Thus, $\max\{\delta,\delta^2\}=\delta=C_{10}\sqrt{\frac{p}{n}}$ for sufficiently large $n$ and
\begin{align*}
P_{\Sigma^*}\left( \|  \hat{\Sigma}  -\Sigma^*\|\geq C_{10} \sqrt{\frac{p}{n}}\|\Sigma^*\|\right)=P_{\Sigma^*}(\|\hat{\Sigma} -\Sigma^*\|\geq \max\{\delta,\delta^2\}\|\Sigma^*\|)\leq\exp(-C_{11} p) \rightarrow 0,
\end{align*}
as $n\rightarrow \infty$.

Since $\epsilon'=\frac{\max\{\sqrt{p},\sqrt{sK^*\log(p+1)}\}}{\sqrt{n}}\|\Sigma^*\|\geq \sqrt{\frac{p}{n}}\|\Sigma^*\|$, we have
\begin{align}
E_{Z^*}\phi=P_{\Sigma^*}\left( \| \hat{\Sigma} -\Sigma^*\|> \frac{C' \epsilon'}{2}\right) \leq 
P_{\Sigma^*}\left( \| \hat{\Sigma} -\Sigma^*\| \geq
 C_{10} \sqrt{\frac{p}{n}}\|\Sigma^*\|\right) \rightarrow 0\label{Lim2}
\end{align}
as $n \rightarrow \infty$, for $C' \geq 2C_{10}$.

Combining results in (\ref{Lim1}) and (\ref{Lim2}), by Lemma \ref{Lem1}, we have 
\begin{align*}
E_{Z^{*}}[\Pi(\|ZZ^T-Z^{*}Z^{*T}\| \geq C' \epsilon' \mid X)]\rightarrow 0, \quad \text{as $n \rightarrow \infty$},
\end{align*}
for $C' \geq 2\max\{C_8, C_{10} \}$. By choosing $C \geq C' C_2$ we have

\begin{align*}
E_{Z^{*}}[\Pi(\|ZZ^T-Z^{*}Z^{*T}\| \geq C \epsilon \mid X)]\rightarrow 0, \quad \text{as $n \rightarrow \infty$}.
\end{align*}
The proof for $Z^* \in \sobm_{p \times K^{*}}$ is now complete.

Next, we consider $Z^{*} = \bm 0$, in which situation  $\epsilon=\frac{\max\{\sqrt{p},\sqrt{\log (p+1)}\}}{\sqrt{n}}$.

Lemma \ref{lem1} ensures that
\begin{align*}
\Pi(Z=Z^{*})\geq \exp(-C\log(p+1)).
\end{align*}
Similar to the proof for $Z^* \in \sobm_{p \times K^{*}}$, if
$\epsilon\rightarrow 0$ as $n\rightarrow \infty$, then 
\begin{align*}
E_{Z^*}\left[ \Pi\left( \| Z Z^T- Z^{*} Z^{*T} \|\leq C \epsilon \mid X\right) \right] \rightarrow 1,
\end{align*}
as $n\rightarrow \infty$ for some positive constant $C$.
The proof is now complete. 
\end{proof}

\subsection{Proof of Corollary \ref{cor1}}
\begin{proof}[Proof of Corollary \ref{cor1}]
For notational simplicity, we omit the index $n$.
Let $Z^* = (\bm{z}_1, \bm{z}_2, \cdots, \bm{z}_p)^{T}$, then $Z^* Z^{*T} = (\bz_a^{T} \bz_b)_{p\times p}$. We have
\begin{align*}
\|Z^*Z^{*T}\|\leq \sup\limits_{1\leq a\leq p}\Big(\sum\limits_{b=1}\limits^{p} \bz_a^{T} \bz_b\Big).
\end{align*}

For each $a$, $\sum\limits_{b=1}\limits^{p} \bz_a^{T} \bz_b$ represents the total number of features shared by object $a$ and all the other objects. If $\bz_a$ has at most $q$ non-zero entries and each column of $Z^*$ has at most $s$ non-zero entries, then
\begin{align*}
\sum\limits_{b=1}\limits^{p} \bz_a^{T} \bz_b\leq sq,
\end{align*}
for every $a$, which implies that
\begin{align*}
\|Z^*Z^{*T}\|\leq sq,
\end{align*}
Since $s,q\geq 1$, we immediately have
\begin{align*}
\epsilon=\frac{\max\{\sqrt{p},\sqrt{sK^*\log(p+1)}\}}{\sqrt{n}}\max\{1,\|Z^*Z^{*T}\|\}\leq \tilde{\epsilon} =\frac{\max\{\sqrt{p},\sqrt{sK^*\log(p+1)}\}}{\sqrt{n}} sq.
\end{align*}
Thus, $\tilde{\epsilon}\rightarrow 0$ implies $\epsilon\rightarrow 0$ (as $n \rightarrow \infty$). According to Theorem \ref{thm1}, we have 
$E_{Z^{*}}[\Pi(\|ZZ^T-Z^{*}Z^{*T}\|\leq C\epsilon \mid X)]\rightarrow 1$. Therefore 
\begin{align*}
E_{Z^{*}}[\Pi(\|ZZ^T-Z^{*}Z^{*T}\| \geq C \tilde{\epsilon} \mid X)] \leq E_{Z^{*}}[\Pi(\|ZZ^T-Z^{*}Z^{*T}\| \geq C \epsilon \mid X)] \rightarrow 0,
\end{align*}
as $n \rightarrow \infty$, i.e. $E_{Z^{*}}[\Pi(\|ZZ^T-Z^{*}Z^{*T}\| \leq C \tilde{\epsilon} \mid X)] \rightarrow 1$,
meaning that $\tilde{\epsilon}$ is the posterior contraction rate given $\tilde{\epsilon}\rightarrow 0$ as $n\rightarrow 0$. Specifically,
\begin{enumerate}
	\item if there is no contraint on the total number of active features in each row, then we set $q=K^*$, meaning that $\tilde{\epsilon}=\frac{\max\{\sqrt{p},\sqrt{sK^*\log(p+1)}\}}{\sqrt{n}} sK^*$ is a valid posterior contraction rate given $\tilde{\epsilon} \rightarrow 0$ as $n\rightarrow \infty$;
	\item if $q$ is bounded or fixed, then $\tilde{\epsilon}=\frac{\max\{\sqrt{p},\sqrt{sK^*\log(p+1)}\}}{\sqrt{n}} s$ is a valid posterior contraction rate given $\tilde{\epsilon} \rightarrow 0$ as $n\rightarrow \infty$.	
\end{enumerate}
\end{proof}
\clearpage

\subsection{Supplementary Tables}

\begin{center}
\begin{longtable}{ | c | c | c | c | c | c | c |  }
\caption{Genes (gene symbol:entrez id) associated with the top 10 proteins with the largest loadings for the 5 most popular features. A $\checkmark$ means the gene is included in the feature.} 
\label{table:app:gene-feature} \\
\hline
Gene      &  1  &  2  &  3  &  4  &  5  & Count \\ \hline
TSC2:7249 & \checkmark  &   & & \checkmark & & 2\\ \hline
DIRAS3:9077 & & \checkmark & \checkmark & & & 2 \\ \hline
BAX:581 & & & \checkmark & \checkmark & & 2 \\ \hline
GSK3A:2931 & & & \checkmark &  & \checkmark & 2 \\ \hline
AKT1:207 & & &  & \checkmark & \checkmark & 2 \\ \hline
ABL1:25  & \checkmark  &   & &  & & 1 \\ \hline
ARAF:369 & \checkmark  &   & &  & & 1 \\ \hline
BID:637  & \checkmark  &   & &  & & 1 \\ \hline
CASP7:840     & \checkmark  &   & &  & & 1 \\ \hline
DIABLO:56616  & \checkmark  &   & &  & & 1 \\ \hline
EEF2K:29904   & \checkmark  &   & &  & & 1 \\ \hline
EIF4EBP1:1978 & \checkmark  &   & &  & & 1 \\ \hline
RB1:5925      & \checkmark  &   & &  & & 1 \\ \hline
RPS6KB1:6198  & \checkmark  &   & &  & & 1 \\ \hline
JAK2:3717     &   & \checkmark  & &  & & 1 \\ \hline
MYH9:4627     &   & \checkmark  & &  & & 1 \\ \hline
NRAS:4893     &   & \checkmark  & &  & & 1 \\ \hline
PREX1:57580   &   & \checkmark  & &  & & 1 \\ \hline
PTEN:5728     &   & \checkmark  & &  & & 1 \\ \hline
SRSF1:6426    &   & \checkmark  & &  & & 1 \\ \hline
STAT3:6774    &   & \checkmark  & &  & & 1 \\ \hline
STAT5A:6776   &   & \checkmark  & &  & & 1 \\ \hline
STMN1:3925    &   & \checkmark  & &  & & 1 \\ \hline
ANXA7:310     &   &  & \checkmark &  & & 1 \\ \hline
COPS5:10987   &   &  & \checkmark &  & & 1 \\ \hline
ERCC5:2073    &   &  & \checkmark &  & & 1 \\ \hline
IRF1:3659     &   &  & \checkmark &  & & 1 \\ \hline
ITGA2:3673    &   &  & \checkmark &  & & 1 \\ \hline
TFRC:7037     &   &  & \checkmark &  & & 1 \\ \hline
YWHAB:7529    &   &  & \checkmark &  & & 1 \\ \hline
BRD4:23476    &   &  &  & \checkmark & & 1 \\ \hline
EIF4G1:1981   &   &  &  & \checkmark & & 1 \\ \hline
MAPK9:5601    &   &  &  & \checkmark & & 1 \\ \hline
MRE11:4361    &   &  &  & \checkmark & & 1 \\ \hline
NFKB1:4790    &   &  &  & \checkmark & & 1 \\ \hline
PRKCA:5578    &   &  &  & \checkmark & & 1 \\ \hline
RBM15:64783   &   &  &  & \checkmark & & 1 \\ \hline
BCL2:596      &   &  &  &  &\checkmark & 1 \\ \hline
BCL2L11:10018 &   &  &  &  &\checkmark & 1 \\ \hline
ERRFI1:54206  &   &  &  &  &\checkmark & 1 \\ \hline
MAP2K1:5604   &   &  &  &  &\checkmark & 1 \\ \hline
MAPK1:5594    &   &  &  &  &\checkmark & 1 \\ \hline
YBX1:4904     &   &  &  &  &\checkmark & 1 \\ \hline
\end{longtable}
\end{center}

\begin{landscape}
\begin{longtable}{ | c | >{\centering\arraybackslash}p{2.5cm} | >{\centering\arraybackslash}p{3cm} | >{\centering\arraybackslash}p{3cm} | >{\centering\arraybackslash}p{3cm} | >{\centering\arraybackslash}p{3cm} | >{\centering\arraybackslash}p{3cm} | c | }
\caption{Pathways that the top genes of the first 5 features are enriched in. Each cell shows the genes (gene symbol:entrez id) of each feature that are enriched in the corresponding pathway.} 
\label{table:app:gene-pathway} \\
\hline
Group   & Pathways & 1  &  2  &  3  &  4  &  5  & Count \\ \hline
Cancer & Pathways in cancer & 
ABL1:25, ARAF:369, RB1:5925, BID:637 & 
NRAS:4893, PTEN:5728, STAT3:6774, STAT5A:6776 & 
ITGA2:3673, BAX:581 & 
AKT1:207, NFKB1:4790, PRKCA:5578, MAPK9:5601, BAX:581 & 
AKT1:207, MAPK1:5594, MAP2K1:5604, BCL2:596 & 
5  \\ \hline
PI3K & PI3K-Akt signaling pathway & 
EIF4EBP1:1978, RPS6KB1:6198, TSC2:7249 & 
JAK2:3717, NRAS:4893, PTEN:5728 & 
ITGA2:3673, YWHAB:7529 &
AKT1:207, NFKB1:4790, PRKCA:5578, TSC2:7249 &
BCL2L11:10018, AKT1:207, MAPK1:5594, MAP2K1:5604, BCL2:596 &
5  \\ \hline
Cancer & Prostate cancer & 
ARAF:369, RB1:5925 & 
NRAS:4893, PTEN:5728 & 
-- & 
AKT1:207, NFKB1:4790 & 
AKT1:207, MAPK1:5594, MAP2K1:5604, BCL2:596 &
4  \\ \hline
Cancer & Chronic myeloid leukemia & 
ABL1:25, ARAF:369, RB1:5925 & 
NRAS:4893, STAT5A:6776 & 
-- & 
AKT1:207, NFKB1:4790 & 
AKT1:207, MAPK1:5594, MAP2K1:5604 & 
4  \\ \hline
Cancer & Glioma & 
ARAF:369, RB1:5925 & 
NRAS:4893, PTEN:5728 & 
-- & 
AKT1:207, PRKCA:5578 & 
AKT1:207, MAPK1:5594, MAP2K1:5604 & 
4  \\ \hline
Cancer & Non-small cell lung cancer & 
ARAF:369, RB1:5925 & 
NRAS:4893, STAT3:6774, STAT5A:6776 & 
-- & 
AKT1:207, PRKCA:5578 & 
AKT1:207, MAPK1:5594, MAP2K1:5604 & 
4  \\ \hline
Cancer & 
Acute myeloid leukemia & 
EIF4EBP1:1978, ARAF:369, RPS6KB1:6198 & 
NRAS:4893, STAT3:6774, STAT5A:6776 & 
-- &
AKT1:207, NFKB1:4790 & 
AKT1:207, MAPK1:5594, MAP2K1:5604 & 
4  \\ \hline
PI3K	 & 
mTOR signaling pathway & 
EIF4EBP1:1978, RPS6KB1:6198, TSC2:7249 & 
NRAS:4893, PTEN:5728 & 
-- & 
AKT1:207, PRKCA:5578, TSC2:7249 & 
AKT1:207, MAPK1:5594, MAP2K1:5604 & 
4  \\ \hline
PI3K	 & 
ErbB signaling pathway	 & 
EIF4EBP1:1978, ABL1:25, ARAF:369, RPS6KB1:6198 & 
NRAS:4893, STAT5A:6776 & 
-- & 
AKT1:207, PRKCA:5578, MAPK9:5601	 & 
AKT1:207, MAPK1:5594, MAP2K1:5604 & 
4  \\ \hline
-- & Non-alcoholic fatty liver disease (NAFLD) & 
BID:637, CASP7:840 & 
-- & 
GSK3A:2931, BAX:581 & 
AKT1:207, NFKB1:4790, MAPK9:5601, BAX:581 & 
BCL2L11:10018, AKT1:207, GSK3A:2931 & 
4  \\ \hline
-- & Regulation of autophagy & 
RPS6KB1:6198, TSC2:7249 & 
NRAS:4893, PTEN:5728 & 
-- & 
AKT1:207, MAPK9:5601, TSC2:7249 & 
AKT1:207, MAPK1:5594, MAP2K1:5604, BCL2:596 &
4   \\ \hline
-- & 
Proteoglycans in cancer & 
ARAF:369, RPS6KB1:6198 & 
NRAS:4893, STAT3:6774 & 
-- & 
AKT1:207, PRKCA:5578 & 
AKT1:207, MAPK1:5594, MAP2K1:5604 &
4   \\ \hline
-- & Hepatitis B & 
-- & 
NRAS:4893, PTEN:5728, STAT3:6774, STAT5A:6776 & 
BAX:581, YWHAB:7529 & 
AKT1:207, NFKB1:4790, PRKCA:5578, MAPK9:5601, BAX:581 & 
AKT1:207, MAPK1:5594, MAP2K1:5604, BCL2:596 & 
4   \\ \hline
Cancer & Melanoma & 
ARAF:369, RB1:5925 & 
NRAS:4893, PTEN:5728 & 
-- & 
-- & 
AKT1:207, MAPK1:5594, MAP2K1:5604 & 
3   \\ \hline
Cancer & Pancreatic cancer & 
ARAF:369, RB1:5925 & 
-- & 
-- & 
AKT1:207, NFKB1:4790, MAPK9:5601 & 
AKT1:207, MAPK1:5594, MAP2K1:5604 & 
3   \\ \hline
PI3K & 	Insulin signaling pathway & 
EIF4EBP1:1978, ARAF:369, RPS6KB1:6198, TSC2:7249 & 
-- & 
-- & 
AKT1:207, MAPK9:5601, TSC2:7249 & 
AKT1:207, MAPK1:5594, MAP2K1:5604 & 
3   \\ \hline
PI3K	 & Apoptosis & 
DIABLO:56616, BID:637, CASP7:840 & 
-- & 
-- & 
AKT1:207, NFKB1:4790, MAPK9:5601, BAX:581 & 
BCL2L11:10018, AKT1:207, MAPK1:5594, MAP2K1:5604, BCL2:596 &
3   \\ \hline
PI3K		 & Chemokine signaling pathway	 & 
--	 & 
JAK2:3717, NRAS:4893, PREX1:57580, STAT3:6774	 & 
--	 & 
AKT1:207, NFKB1:4790	 & 
AKT1:207, GSK3A:2931, MAPK1:5594, MAP2K1:5604	 & 
3   \\ \hline
PI3K &
FoxO signaling pathway &
-- &
NRAS:4893, PTEN:5728, STAT3:6774 &
-- &
AKT1:207, MAPK9:5601 &
BCL2L11:10018, AKT1:207, MAPK1:5594, MAP2K1:5604 &
3   \\ \hline
PI3K & 
MAPK signaling pathway & 
-- & 
STMN1:3925, NRAS:4893 & 
-- & 
AKT1:207, NFKB1:4790, PRKCA:5578, MAPK9:5601 & 
AKT1:207, MAPK1:5594, MAP2K1:5604 & 
3   \\ \hline
-- & HIF-1 signaling pathway & 
EIF4EBP1:1978, RPS6KB1:6198 & 
-- & 
-- & 
AKT1:207, NFKB1:4790, PRKCA:5578 & 
AKT1:207, MAPK1:5594, MAP2K1:5604, BCL2:596 & 
3   \\ \hline
-- & Choline metabolism in cancer & 
EIF4EBP1:1978, RPS6KB1:6198, TSC2:7249 & 
-- & 
-- & 
AKT1:207, PRKCA:5578, MAPK9:5601, TSC2:7249 & 
AKT1:207, MAPK1:5594, MAP2K1:5604 & 
3   \\ \hline
-- & Viral carcinogenesis & 
-- & 
NRAS:4893, STAT3:6774, STAT5A:6776 & 
BAX:581, YWHAB:7529 & 
NFKB1:4790, BAX:581 & 
-- & 
3   \\ \hline
-- & Cholinergic synapse & 
-- & 
JAK2:3717, NRAS:4893 & 
-- & 
AKT1:207, PRKCA:5578 & 
AKT1:207, MAPK1:5594, MAP2K1:5604, BCL2:596 & 
3   \\ \hline
-- & 
Prolactin signaling pathway & 
-- & 
JAK2:3717, NRAS:4893, STAT3:6774, STAT5A:6776 & 
-- & 
AKT1:207, NFKB1:4790, MAPK9:5601 & 
AKT1:207, MAPK1:5594, MAP2K1:5604 & 
3   \\ \hline
-- & Toxoplas-mosis & 
-- & 
JAK2:3717, STAT3:6774 & 
-- & 
AKT1:207, NFKB1:4790, MAPK9:5601 & 
AKT1:207, MAPK1:5594, BCL2:596 & 
3   \\ \hline
-- & Sphingolipid signaling pathway & 
-- & 
NRAS:4893, PTEN:5728 & 
-- & 
AKT1:207, NFKB1:4790, PRKCA:5578, MAPK9:5601, BAX:581 & 
AKT1:207, MAPK1:5594, MAP2K1:5604, BCL2:596 & 
3   \\ \hline
-- & MicroRNAs in cancer & 
-- & 
STMN1:3925, NRAS:4893, PTEN:5728, STAT3:6774 & 
-- & 
NFKB1:4790, PRKCA:5578 & 
BCL2L11:10018, MAPK1:5594, MAP2K1:5604, BCL2:596 & 
3   \\ \hline
-- & Hepatitis C & 
-- & 
NRAS:4893, STAT3:6774 & 
-- & 
AKT1:207, NFKB1:4790, MAPK9:5601 & 
AKT1:207, MAPK1:5594 & 
3   \\ \hline
Cancer & 
Bladder cancer & 
ARAF:369, RB1:5925 & 
-- & 
-- & 
-- & 
MAPK1:5594, MAP2K1:5604 & 
2   \\ \hline
Cancer & 
Endometrial cancer & 
-- & 
NRAS:4893, PTEN:5728 & 
-- & 
-- & 
AKT1:207, MAPK1:5594, MAP2K1:5604 & 
2   \\ \hline
Cancer & 
Colorectal cancer & 
-- & 
-- & 
-- & 
AKT1:207, MAPK9:5601, BAX:581 & 
AKT1:207, MAPK1:5594, MAP2K1:5604, BCL2:596 & 
2   \\ \hline
Cancer & 
Small cell lung cancer & 
-- & 
-- & 
-- & 
AKT1:207, NFKB1:4790 & 
AKT1:207, BCL2:596 & 
2   \\ \hline
PI3K	 & 
p53 signaling pathway & 
BID:637, TSC2:7249 & 
-- & 
-- & 
BAX:581, TSC2:7249 & 
-- & 
2   \\ \hline
PI3K & 
Jak-STAT signaling pathway & 
-- & 
JAK2:3717, STAT3:6774, STAT5A:6776 & 
-- & 
-- & 
AKT1:207, BCL2:596 & 
2   \\ \hline
PI3K & 
Focal adhesion & 
-- & 
-- & 
-- & 
AKT1:207, PRKCA:5578, MAPK9:5601 & 
AKT1:207, MAPK1:5594, MAP2K1:5604, BCL2:596 & 
2   \\ \hline
PI3K & 
B cell receptor signaling pathway & 
-- & 
-- & 
-- & 
AKT1:207, NFKB1:4790 & 
AKT1:207, MAPK1:5594, MAP2K1:5604 & 
2   \\ \hline
PI3K & 
Toll-like receptor signaling pathway & 
-- & 
-- & 
-- & 
AKT1:207, NFKB1:4790, MAPK9:5601 & 
AKT1:207, MAPK1:5594, MAP2K1:5604 & 
2   \\ \hline
PI3K	 & 
VEGF signaling pathway & 
-- & 
-- & 
-- & 
AKT1:207, PRKCA:5578 & 
AKT1:207, MAPK1:5594, MAP2K1:5604 & 
2   \\ \hline
-- & 
AMPK signaling pathway & 
EIF4EBP1:1978, EEF2K:29904, RPS6KB1:6198, TSC2:7249 & 
-- & 
-- & 
AKT1:207, TSC2:7249 & 
-- & 
2   \\ \hline
-- & 
Natural killer cell mediated cytotoxicity & 
ARAF:369, BID:637 & 
-- & 
-- & 
-- & 
MAPK1:5594, MAP2K1:5604 & 
2   \\ \hline
-- & HTLV-I infection & 
-- & 
NRAS:4893, STAT5A:6776 & 
-- & 
AKT1:207, NFKB1:4790, BAX:581 & 
-- & 
2   \\ \hline
-- & 
Measles & 
-- & 
JAK2:3717, STAT3:6774, STAT5A:6776 & 
-- & 
AKT1:207, NFKB1:4790 & 
-- & 
2   \\ \hline
-- & Herpes simplex infection & -- & JAK2:3717, SRSF1:6426 & -- & NFKB1:4790, MAPK9:5601 & --
& 2   \\ \hline
-- & Adipocytokine signaling pathway & -- & JAK2:3717, STAT3:6774 & -- & AKT1:207, NFKB1:4790, MAPK9:5601 & --
& 2   \\ \hline
-- & Signaling pathways regulating pluripotency of stem cells & -- & JAK2:3717, NRAS:4893, STAT3:6774 & -- & -- & AKT1:207, MAPK1:5594, MAP2K1:5604
& 2   \\ \hline
-- & Regulation of actin cytoskeleton & -- & MYH9:4627, NRAS:4893 & -- & -- & MAPK1:5594, MAP2K1:5604
& 2   \\ \hline
-- & Central carbon metabolism in cancer & -- & NRAS:4893, PTEN:5728 & -- & -- & AKT1:207, MAPK1:5594, MAP2K1:5604
& 2   \\ \hline
-- & 
Epstein-Barr virus infection & 
-- & 
-- & 
-- & 
AKT1:207, NFKB1:4790, MAPK9:5601 & 
AKT1:207, BCL2:596 & 
2   \\ \hline
-- & 
TNF signaling pathway & 
-- & 
-- & 
-- & 
AKT1:207, NFKB1:4790, MAPK9:5601 & 
AKT1:207, MAPK1:5594, MAP2K1:5604 & 
2   \\ \hline
-- & Tuberculosis & -- & -- & -- & AKT1:207, NFKB1:4790, MAPK9:5601, BAX:581 & AKT1:207, MAPK1:5594, BCL2:596 & 
2   \\ \hline
-- & Neurotrophin signaling pathway & -- & -- & -- & AKT1:207, NFKB1:4790, MAPK9:5601, BAX:581 & AKT1:207, MAPK1:5594, MAP2K1:5604, BCL2:596 & 
2   \\ \hline
-- & Ras signaling pathway & -- & -- & -- & AKT1:207, NFKB1:4790, PRKCA:5578, MAPK9:5601 & AKT1:207, MAPK1:5594, MAP2K1:5604 & 
2   \\ \hline
-- & Adrenergic signaling in cardiomyocytes & -- & -- & -- & AKT1:207, PRKCA:5578 & AKT1:207, MAPK1:5594, BCL2:596 & 
2   \\ \hline
-- & Thyroid hormone signaling pathway & -- & -- & -- & AKT1:207, PRKCA:5578, TSC2:7249 & AKT1:207, MAPK1:5594, MAP2K1:5604 &
2   \\ \hline
-- & Progesterone-mediated oocyte maturation & -- & -- & -- & AKT1:207, MAPK9:5601 & AKT1:207, MAPK1:5594, MAP2K1:5604
& 2 \\  \hline
-- & Chagas disease (American trypanosomiasis) & -- & -- & -- & AKT1:207, NFKB1:4790, MAPK9:5601 & AKT1:207, MAPK1:5594
& 2 \\  \hline
-- & Osteoclast differentiation & -- & -- & -- & AKT1:207, NFKB1:4790, MAPK9:5601 & AKT1:207, MAPK1:5594, MAP2K1:5604
& 2 \\  \hline
-- & T cell receptor signaling pathway & -- & -- & -- & AKT1:207, NFKB1:4790, MAPK9:5601 & AKT1:207, MAPK1:5594, MAP2K1:5604
& 2 \\  \hline
-- & cAMP signaling pathway & -- & -- & -- & AKT1:207, NFKB1:4790, MAPK9:5601 & AKT1:207, MAPK1:5594, MAP2K1:5604
& 2 \\  \hline
-- & Influenza A & -- & -- & -- & AKT1:207, NFKB1:4790, PRKCA:5578, MAPK9:5601 & AKT1:207, MAPK1:5594, MAP2K1:5604
& 2 \\  \hline
-- & Fc gamma R-mediated phagocytosis & -- & -- & -- & AKT1:207, PRKCA:5578 & AKT1:207, MAPK1:5594, MAP2K1:5604
& 2 \\  \hline
-- & Rap1 signaling pathway & -- & -- & -- & AKT1:207, PRKCA:5578 & AKT1:207, MAPK1:5594, MAP2K1:5604
& 2 \\  \hline
-- & Dopaminergic synapse & -- & -- & -- & AKT1:207, PRKCA:5578, MAPK9:5601 & AKT1:207, GSK3A:2931
& 2 \\  \hline
-- & Fc epsilon RI signaling pathway & -- & -- & -- & AKT1:207, PRKCA:5578, MAPK9:5601 & AKT1:207, MAPK1:5594, MAP2K1:5604
& 2 \\  \hline
-- & NOD-like receptor signaling pathway & -- & -- & -- & NFKB1:4790, MAPK9:5601 & MAPK1:5594, BCL2:596
& 2 \\  \hline
-- & GnRH signaling pathway & -- & -- & -- & PRKCA:5578, MAPK9:5601 & MAPK1:5594, MAP2K1:5604
& 2 \\  \hline
Cancer & Renal cell carcinoma & -- & -- & -- & -- & AKT1:207, MAPK1:5594, MAP2K1:5604
& 1 \\  \hline
Cancer & Thyroid cancer & -- & -- & -- & -- & MAPK1:5594, MAP2K1:5604
& 1 \\  \hline
PI3K & Cell cycle & ABL1:25, RB1:5925 & -- & -- & -- & --
& 1 \\  \hline
-- & Viral myocarditis & ABL1:25, BID:637 & -- & -- & -- & --
& 1 \\  \hline
-- & Alzheimer's disease & BID:637, CASP7:840 & -- & -- & -- & --
& 1 \\  \hline
-- & Hematopoietic cell lineage & -- & -- & ITGA2:3673, TFRC:7037 & -- & --
& 1 \\  \hline
-- & Phagosome & -- & -- & ITGA2:3673, TFRC:7037 & -- & --
& 1 \\  \hline
-- & Protein processing in endoplasmic reticulum & -- & -- & -- & MAPK9:5601, BAX:581 & --
& 1 \\  \hline
-- & Epithelial cell signaling in Helicobacter pylori infection & -- & -- & -- & NFKB1:4790, MAPK9:5601 & --
& 1 \\  \hline
-- & Pertussis & -- & -- & -- & NFKB1:4790, MAPK9:5601 & --
& 1 \\  \hline
-- & RIG-I-like receptor signaling pathway & -- & -- & -- & NFKB1:4790, MAPK9:5601 & --
& 1 \\  \hline
-- & Salmonella infection & -- & -- & -- & NFKB1:4790, MAPK9:5601 & --
& 1 \\  \hline
-- & Shigellosis & -- & -- & -- & NFKB1:4790, MAPK9:5601 & --
& 1 \\  \hline
-- & Amoebiasis & -- & -- & -- & NFKB1:4790, PRKCA:5578 & --
& 1 \\  \hline
-- & Inflammatory mediator regulation of TRP channels & -- & -- & -- & PRKCA:5578, MAPK9:5601 & --
& 1 \\  \hline
-- & Retrograde endocannabinoid signaling & -- & -- & -- & PRKCA:5578, MAPK9:5601 & --
& 1 \\  \hline
-- & Wnt signaling pathway & -- & -- & -- & PRKCA:5578, MAPK9:5601 & --
& 1 \\  \hline
-- & Platelet activation & -- & -- & -- & -- & AKT1:207, MAPK1:5594
& 1 \\  \hline
-- & cGMP-PKG signaling pathway & -- & -- & -- & -- & AKT1:207, MAPK1:5594, MAP2K1:5604
& 1 \\  \hline
-- & Alcoholism & -- & -- & -- & -- & MAPK1:5594, MAP2K1:5604
& 1 \\  \hline
-- & Dorso-ventral axis formation & -- & -- & -- & -- & MAPK1:5594, MAP2K1:5604
& 1 \\  \hline
-- & Gap junction & -- & -- & -- & -- & MAPK1:5594, MAP2K1:5604
& 1 \\  \hline
-- & Long-term depression & -- & -- & -- & -- & MAPK1:5594, MAP2K1:5604
& 1 \\  \hline
-- & Long-term potentiation & -- & -- & -- & -- & MAPK1:5594, MAP2K1:5604
& 1 \\  \hline
-- & Melanogenesis & -- & -- & -- & -- & MAPK1:5594, MAP2K1:5604
& 1 \\  \hline
-- & Oxytocin signaling pathway & -- & -- & -- & -- & MAPK1:5594, MAP2K1:5604
& 1 \\  \hline
-- & Prion diseases & -- & -- & -- & -- & MAPK1:5594, MAP2K1:5604
& 1 \\  \hline
-- & Serotonergic synapse & -- & -- & -- & -- & MAPK1:5594, MAP2K1:5604
& 1 \\  \hline
-- & Vascular smooth muscle contraction & -- & -- & -- & -- & MAPK1:5594, MAP2K1:5604
& 1 \\  \hline
-- & Estrogen signaling pathway & -- & -- & -- & -- & AKT1:207, MAPK1:5594, MAP2K1:5604
& 1 \\  \hline
-- & Oocyte meiosis & -- & -- & -- & -- & MAPK1:5594, MAP2K1:5604
& 1 \\  \hline
\end{longtable}
\end{landscape}

\end{document}